\documentclass[runningheads,a4paper]{llncs}
\usepackage[top=1.3in,bottom=1.3in, left=1.2in, right=1.2in]{geometry}

\usepackage{amssymb}
\usepackage{amsmath}
\usepackage{color}
\usepackage{subfigure}
\usepackage{graphicx}
\usepackage[numbers]{natbib}
\usepackage{url}

\newcommand{\be}{\begin{equation}}
\newcommand{\ee}{\end{equation}}
\newcommand{\setn}{\{1, \ldots, n\}}
\newcommand{\setwv}{[w_1, v_1] \times \ldots \times [w_n, v_n]}

\DeclareMathOperator*{\argmax}{arg\,max}

\begin{document}

\mainmatter  

\title{Revenue-Maximizing Mechanism Design for Quasi-Proportional Auctions}

\titlerunning{Revenue-Maximizing Mechanism Design for Quasi-Proportional Auctions}

%
%
\author{Zheng Wen 
\and Eric Bax
\and James Li}
\authorrunning{Wen et al.}

\urldef{\mailsa}\path|zhengwen, ebax, jamesyili@yahoo-inc.com|
\institute{Yahoo Labs\\ \mailsa}

%
%

\toctitle{Lecture Notes in Computer Science}
\tocauthor{Authors' Instructions}
\maketitle

\begin{abstract}
In quasi-proportional auctions, each bidder receives a fraction of the allocation equal to the weight of their bid divided by the sum of weights of all bids, where each bid's weight is determined by a weight function. We study the relationship between the weight function, bidders' private values, number of bidders, and the seller's revenue in equilibrium. It has been shown that if one bidder has a much higher private value than the others, then a nearly flat weight function maximizes revenue. Essentially, threatening the bidder who has the highest valuation with having to share the allocation maximizes the revenue. We show that as bidder private values approach parity, steeper weight functions maximize revenue by making the quasi-proportional auction more like a winner-take-all auction. We also show that steeper weight functions maximize revenue as the number of bidders increases. For flatter weight functions, there is known to be a unique pure-strategy Nash equilibrium. We show that a pure-strategy Nash equilibrium also exists for steeper weight functions, and we give lower bounds for bids at an equilibrium. For a special case that includes the two-bidder auction, we show that the pure-strategy Nash equilibrium is unique, and we show how to compute the revenue at equilibrium. We also show that selecting a weight function based on private value ratios and number of bidders is necessary for a quasi-proportional auction to produce more revenue than a second-price auction.

\end{abstract}

\section{Introduction}
\label{sec_introduction}

Quasi-proportional auctions \cite{tullock80,kelly97} award each bidder a fraction of the total allocation equal to the weight of their bid divided by the sum of weights of all bids, where a bid's weight is determined by a weight function. Hence, the allocation for bidder $i$ is
\be
a_i(\mathbf{b}) =  \frac{f(b_i)}{\sum_j f(b_j)},
\ee
where $\mathbf{b}$ is the vector of bids and $f$ is the weight function. In this paper, we focus on winners-pay quasi-proportional auctions, in which bidders pay their bid times their allocation. (A well-known alternative is the all-pay auction, in which all bidders pay their full bid regardless of allocation \cite{tullock80}. The all-pay auction has been used as a model for disparate interests plying officials with gifts and favors in hopes of influencing policy.)

Why use quasi-proportional auctions? It is well known that the revenue-optimal auction (for a single item, non-repeated) is the second-price auction with optimal reserve prices \cite{myerson81,riley81}. The optimal reserve prices are based on knowledge of prior distributions from which bidders draw their private values. Without this knowledge, we are in a \textit{prior-free} setting \cite{goldberg02,hartline07}, in which it can be a challenge to set effective reserve prices \cite{muthukrishnan09}. If the auction is repeated and priors are stable over time, then the priors may be learnable \cite{li10,cole14,dughmi14,hummel14}. However, in some practical scenarios with unknown priors, either the auction is not repeated, or the priors change from auction to auction. 

In the prior-free setting without reserve prices, Mirrokni et al. \cite{mirrokni10} show that a quasi-proportional auction has better worst-case performance than a second-price auction. In their worst case, the bidder with the highest private value has a much higher private value than the other bidders. In this case, they show that quasi-proportional auctions with functions $f(x) = x^p$ and $p \leq 1$, called Tullock auctions \cite{tullock80}, can achieve $\Omega$($\sqrt{\alpha}$) revenue, where $\alpha$ is the ratio of the highest private value to the next-highest private value. These results are called prior-free revenue results. 
Nguyen and Vojnovic \cite{nguyen10} show that for the prior-free setting without reserve prices, there is an upper bound of o($\frac{v_1}{\log(\frac{v_1}{v_2})}$) on equilibrium revenue, where $v_1$ and $v_2$ are the two highest private values of bidders. They also give a mechanism that achieves revenue $\Omega$($\frac{v_1}{\log^{1+\epsilon}(\frac{v_1}{v_2} + 1)}$), which is similar to the best known prior-free result with reserve prices \cite{lu06}.

Other than a lack of priors, some reasons why a quasi-proportional allocation might be useful include:

\begin{enumerate}
\item \textbf{The item is always awarded.} For auctions with reserve prices, the item may be withheld from all bidders. This may create a problem for the seller. Quasi-proportional auctions avoid this problem. 
\item \textbf{There is a shared allocation.} The seller may desire a shared allocation if a zero allocation to runner-up bidders makes them unlikely to participate in future auctions, which can decrease competition and revenue in those auctions. A shared allocation awards a ``second prize for the second price." (A single allocation is also possible: the auctioneer can award the item at random, with each bidder's probability of winning the item equal to its fraction of the allocation \cite{muthukrishnan09}.)
\end{enumerate}  

For the winners-pay quasi-proportional auction, assume bidder $i$ has utility function
\begin{equation}
u_i(\mathbf{b}) = a_i(\mathbf{b}) (v_i - b_i),
\end{equation}
where $v_i$ is bidder $i$'s private value for a full allocation. 
Without loss of generality, in this paper we assume the weight function $f$ has the form $f(x) = x^p$, where the exponent $p > 0$ is to be specified. 
The main contributions of this paper are: 

\begin{enumerate}
\item For all $p>0$, we show that the quasi-proportional auction with weight function $f(x) = x^p$ has a pure-strategy Nash equilibrium.
\item We give lower bounds for bids at a pure-strategy Nash equilibrium. The bounds are based on $p$ and the level of competition, in terms of number of bidders and their private values. As competition increases, higher values of $p$ maximize the bounds.
\item For the case of a bidder with a higher private value and one or more bidders that share a lower private value, we show that the pure-strategy Nash equilibrium is unique.
\item For that case, we show how to compute equilibrium revenue, and we explore how it depends on $p$, the ratio of higher to lower private values, and the number of bidders. We show that steeper weight functions are needed to maximize revenue as competition increases. 
\end{enumerate}



\section{Bid Lower Bounds at an Equilibrium}
\label{sec:general}

In this section, we consider quasi-proportional auctions with weight functions $f(x) = x^p$ for $p>0$. For these auctions, we prove the existence of a pure-strategy Nash equilibrium, and we give lower bounds for bids at an equilibrium. Mirrokni et al. \cite{mirrokni10} show there is a unique pure-strategy Nash equilibrium for $0<p\leq1$; we show that an equilibrium also exists for $p>1$. When $p>1$, a bidder's response function is not necessarily concave over the whole domain of their bids that give nonnegative utility: $[0,v]$, where $v$ is the bidder's private value. Thus, we cannot apply results such as those by Rosen \cite{rosen65}, as Mirrokni et al. do, to show the existence and uniqueness of a pure-strategy Nash equilibrium. Instead, we use Brouwer's fixed-point theorem to prove existence of an equilibrium. 

First we show that each bidder's best response is unique for all $p>0$. Next we derive vectors of lower bounds, $\mathbf{w} \in \Re^{+n}$, such that if all bidders $i$ bid $b_i \geq w_i$, where $w_i$ is the $i$th component of $\mathbf{w}$, then the same holds for the vector of bidders' best responses to each others' bids. We combine that result with Brouwer's fixed-point theorem to prove the existence of a pure-strategy Nash equilibrium. Since the bids at equilibrium are at least as great as the lower bounds $\mathbf{w}$, the lower bounds on bids also imply lower bounds for the auctioneer's equilibrium revenue. 

In this section, we use $b$ to represent a single bidder's bid and $\mathbf{b}$ to represent the vector of bidder's bids; we drop the $(b)$ notation that indicates functions of $b$, for example writing $u$ instead of $u(b)$ for a bidder's response function, and we use apostrophes to denote derivatives with respect to $b$, such as $f''$ for the second derivative of the weight function. 

\subsection{Uniqueness of the Best Response}
To show that each bidder's best response is unique, we will show that each bidder's response curve (their utility curve given other bidders' bids) is concave at all points that have derivative zero. Since the response curves are continuous, this implies there are no local minima. If there were multiple local maxima, there would have to be a local minimum between each successive pair of local maxima. So there can only be a single global maximum. (There is a maximum since the response curve is zero and ascending at zero bid, zero and descending at bid equal to the bidder's private value, and the response curve and its first derivative are continuous. So we can apply Bolzano's theorem to the derivative.)  

We will use $b$ to represent a single bidder's bid. Holding the other bidders' bids fixed, the allocation is 
\be
a = \frac{f}{f + s},
\ee
where $s$ is the sum of weight functions of the other bidders' bids. The bidder's utility function is
\be
u = a (v-b),
\ee
where $v$ is the bidder's private value. We have the following theorem:

\begin{theorem} \label{thm_best_response}
If $f = b^p$ such that $p>0$, and $s > 0$, then $(u' = 0) \Rightarrow (u''<0)$.
\end{theorem}
Please refer to the appendix for the proof.

Figure \ref{u_plot} shows the utility function and its first two derivatives for $p=1$, $p=2$, and $p=4$. In each case, we set $s$, the sum of weight functions for other bids, to 7, and the bidder's private value to 10. In Figure \ref{u_plot_p1}, with $p=1$, the utility function is concave (negative second derivative) over the whole bid domain [0,10]. In Figures \ref{u_plot_p2} and \ref{u_plot_p4}, the utility functions are not concave at $b=0$, but the second derivatives are negative at the $b$ values where the first derivatives are zero. In Figure \ref{u_plot_p1}, the utility function is a smoothly rounded curve. In Figure \ref{u_plot_p2}, the utility function is still somewhat rounded, but there is a slight S-curve starting at zero, and the curve is less rounded and more linear to the right of the maximum. In Figure \ref{u_plot_p4} these effects are more pronounced. The S-curve that makes the utility functions in Figures \ref{u_plot_p2} and \ref{u_plot_p4} non-concave on the left results from slow gains in allocation near bid zero, followed by quick gains due to the growth of $f = b^p$ for large $p$, then by leveling off in the allocation as $b^p$ comes to dominate the denominator in $a = \frac{b^p}{b^p + s}$. 

\begin{figure}
\subfigure[$p=1$]{\includegraphics[width = 0.31\textwidth]{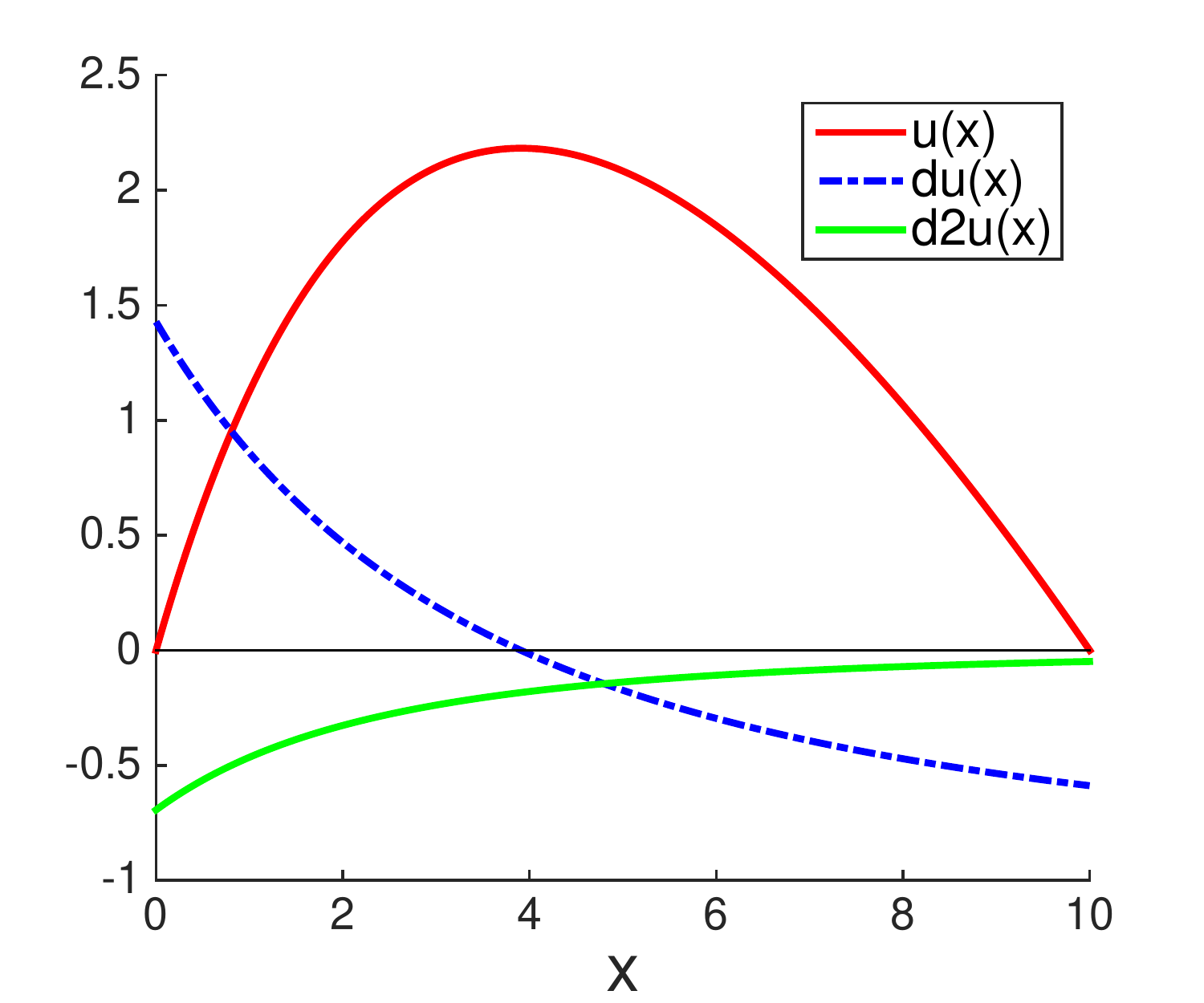} \label{u_plot_p1}}
\subfigure[$p=2$]{\includegraphics[width = 0.31\textwidth]{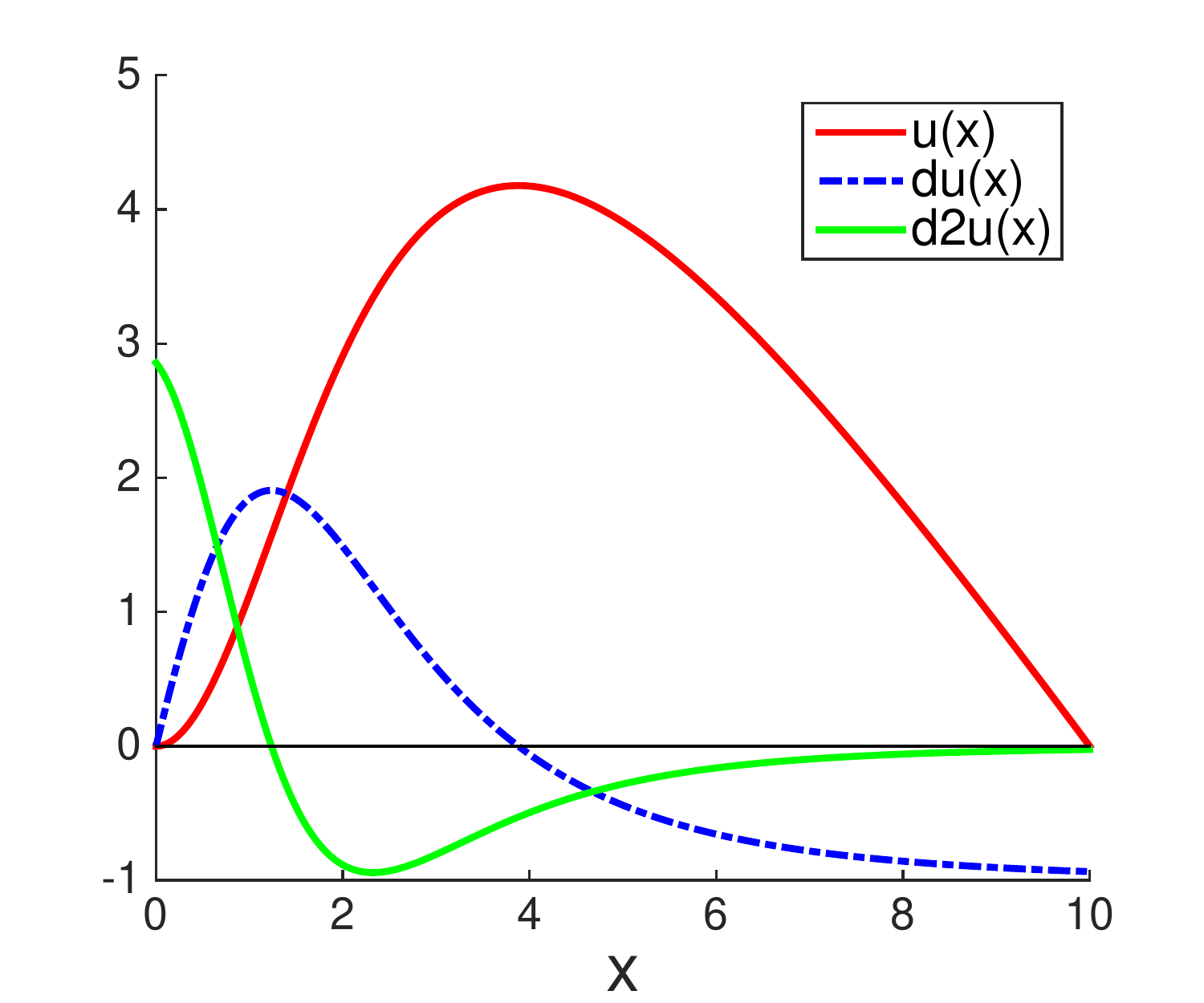} \label{u_plot_p2}}
\subfigure[$p=4$]{\includegraphics[width = 0.31\textwidth]{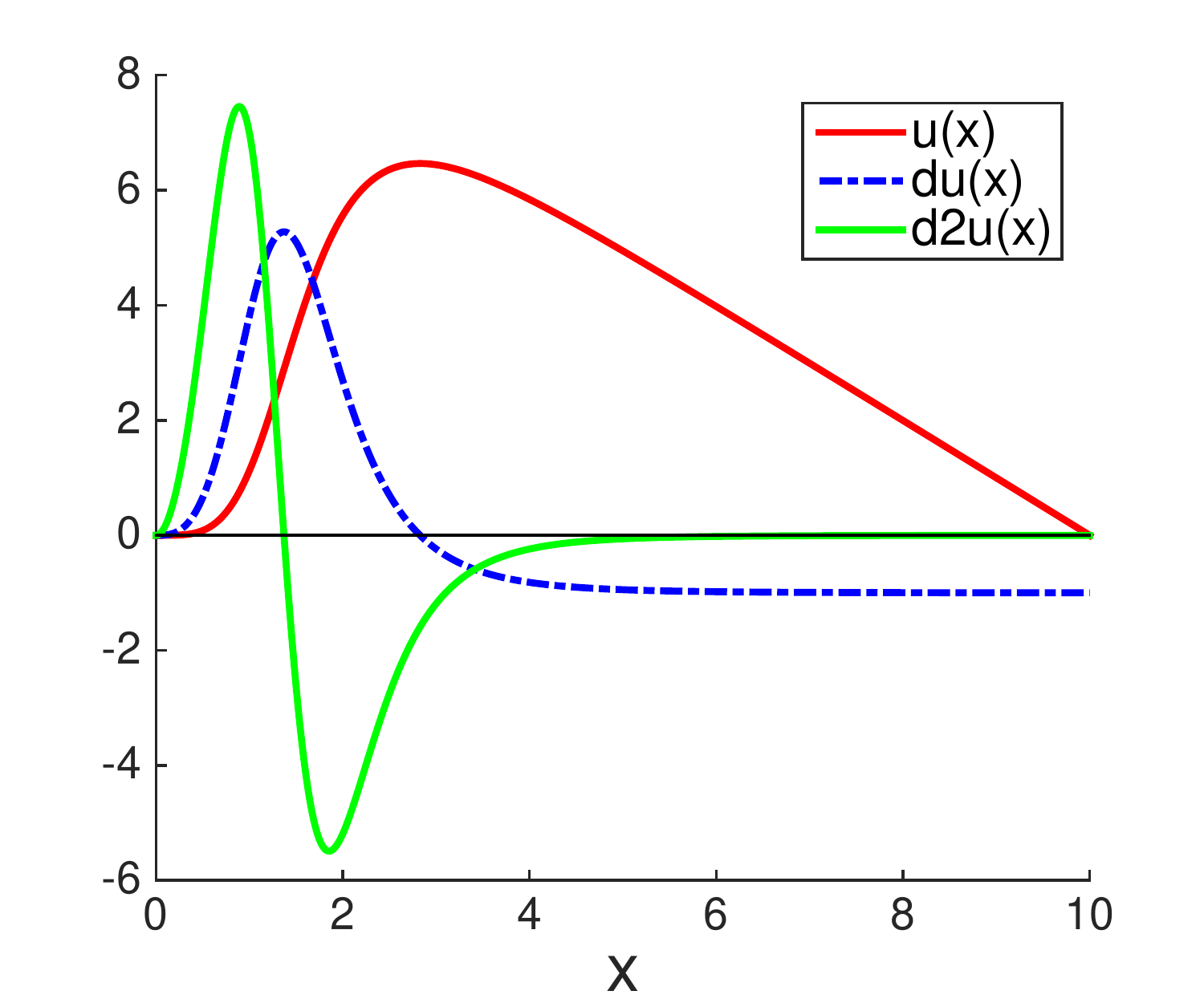} \label{u_plot_p4}}
\caption{Utility Functions and Derivatives for $p=1, 2, 4$}
\label{u_plot}
\end{figure}

%
%

\subsection{Pure-Strategy Nash Equilibria and Bid Lower Bounds}
Let $\mathbf{b} = (b_1, \ldots, b_n)$ be a vector of bids, with $b_i$ the bid for bidder $i$. Let $\mathrm{BR}_i(\mathbf{b})$ be bidder $i$'s best response to the other bidders' bids. Define response function $\mathrm{BR}(\mathbf{b}) = (\mathrm{BR}_1(\mathbf{b}), \ldots, \mathrm{BR}_n(\mathbf{b}))$. Let $v_i$ be the valuation for bidder $i$, and let $f$ be the weight function.
We have the following theorem:
\begin{theorem} \label{eq}
\label{thm_bounded_mapping}
Suppose $f = b^p$ for $p>0$, and a vector $\mathbf{w}$ meets the conditions
\be
\forall i \in \setn: w_i \leq \frac{v_i}{1 + \frac{1}{p}(1+\frac{f(w_i)}{s_i})},  \label{ineq_w}
\ee
where $s_i = \sum_{j \not= i} f(w_j)$. Then, for all $\mathbf{b}$ such that $\forall i: b_i \geq w_i$, $\forall i: \mathrm{BR}_i(\mathbf{b}) \geq w_i$.
\end{theorem}
In other words, if $\mathbf{w}$ meets the conditions of the theorem, then $\mathrm{BR}$ maps $\setwv$ into itself. Please refer to the appendix for the proof.

Before we use this theorem to prove existence of an equilibrium and derive lower bounds for bids at an equilibrium, look at Inequality \ref{ineq_w}. The value of $p$ mediates a tradeoff. From the term $\frac{1}{p}$, the bound gets stronger (closer to $v$) as $p$ increases. However, for the bidder with the greatest private value, increasing $p$ can increase the weight function of that bidder's bid, $f$, so much that it dominates the sum of weight functions of other bidders' bids, $s$, and this effect becomes more pronounced as the ratio between the highest private value and other private values increases. So we can get stronger bounds by using higher $p$ if there is enough competition to keep $\frac{f}{s}$ low.

Using Theorem \ref{thm_bounded_mapping}, it is straightforward to prove that there exist pure-strategy Nash equilibria for quasi-proportional auction mechanisms with convex weight functions. 
\begin{theorem}
For any bounds $\mathbf{w}$ that meet the conditions of Theorem \ref{eq}, there exists a pure-strategy Nash equilibrium in $\setwv$.
\end{theorem}
\begin{proof}
Note that the best response function $\mathrm{BR}$ is a continuous function of $\mathbf{b}$ over $\setwv$. Theorem \ref{eq} implies that $\mathrm{BR}$ maps $\setwv$ into itself. Brouwer's fixed point theorem \cite{brouwer12,franklin02} states that every continuous function from a convex compact set into itself has a fixed point. Therefore,
\be
\exists \mathbf{b}^* \in \setwv: \mathrm{BR}(\mathbf{b}^*) = \mathbf{b}^*.
\ee
\end{proof}

We now present some lower bound $\mathbf{w}$'s satisfying Inequality \ref{ineq_w}.
The following corollary gives lower bounds that are the same for all bidders. 

\begin{corollary}
\label{corollary_1}
Let $v_{\min} = \min(v_1, \ldots, v_n)$. Then Theorem \ref{eq} applies to lower bounds 
\be
\forall i: w_i = \frac{v_{\min}}{1 + \frac{1}{p}(1 + \frac{1}{n-1})} .
\ee
\end{corollary}
\begin{proof}
We will show that the condition of Theorem \ref{eq},
\be
\frac{v_i}{1 + \frac{1}{p}(1+\frac{f(w_i)}{s_i})} \geq w_i, \nonumber
\ee
holds. Since $w_1 = \ldots = w_n$ for these bounds, $\forall i: f(w_i) = \frac{s_i}{n-1}$. So 
\be
\forall i: \frac{v_i}{1 + \frac{1}{p}(1+\frac{f(w_i)}{s_i})} = \frac{v_i}{1 + \frac{1}{p}(1+\frac{1}{n-1})}  \nonumber
\ee
\be
\geq \frac{v_{\min}}{1 + \frac{1}{p}(1+\frac{1}{n-1})} = w_i. \nonumber
\ee
\end{proof}

Note that when $v_1=v_2=\ldots =v_n$, and $p=1$, then Corollary~\ref{corollary_1} gives a lower bound $\frac{v_1}{2 + \frac{1}{n-1}}$, which is $v_1/3$ when $n=2$ and approaches to
$v_1/2$ as $n \rightarrow \infty$. Moreover, for any $n$, this lower bound approaches $v_1$ as $p \rightarrow \infty$. 

Corollary~\ref{corollary_2} offers bounds based on the second-highest valuation. These bounds are stronger than those from Corollary \ref{corollary_1} when the second-highest valuation is significantly higher than the minimum valuation.



\begin{corollary}
\label{corollary_2}
Assume, without loss of generality, $v_1 \geq \ldots \geq v_n$. Let $w_1 = w_2 = \frac{v_2}{1 + \frac{2}{p}}$. For $i>2$, let 
\be
w_i = \min(\frac{v_i}{1 + \frac{1}{p}(1 + \frac{1}{i-1})}, w_{i-1}).
\ee
For these $\mathbf{w}$ values, the response function $\mathrm{BR}$ maps $\setwv$ into itself. 
\end{corollary}
\begin{proof}
We will show that these bounds meet the condition from Theorem \ref{eq}. For each $i$, we need to show
\be
\frac{v_i}{1 + \frac{1}{p}(1+\frac{f(w_i)}{s_i})} \geq w_i. \nonumber
\ee
For $i=1$:
\be
\frac{v_1}{1 + \frac{1}{p}(1+\frac{f(w_1)}{s_1})} \geq \frac{v_1}{1 + \frac{1}{p}(1+\frac{f(w_1)}{f(w_2)})} =  \frac{v_1}{1 + \frac{2}{p}} \geq \frac{v_2}{1 + \frac{2}{p}} = w_1. \nonumber
\ee
For $i=2$:
\be
\frac{v_2}{1 + \frac{1}{p}(1+\frac{f(w_2)}{s_2})} \geq \frac{v_2}{1 + \frac{1}{p}(1+\frac{f(w_2)}{f(w_1)})} =  \frac{v_2}{1 + \frac{2}{p}} = w_2. \nonumber
\ee
For $i>2$:
\be
\frac{v_i}{1 + \frac{1}{p}(1+\frac{f(w_i)}{s_i})} \geq \frac{v_i}{1 + \frac{1}{p}(1+\frac{f(w_i)}{\sum_{j=1}^{i-1} f(w_j)})} \geq \frac{v_i}{1 + \frac{1}{p}(1 + \frac{1}{i-1})} \geq w_i. \nonumber
\ee
\end{proof}

In general, there are many $\mathbf{w} \in \Re^{+n}$ satisfying Inequality \ref{ineq_w}, and each of them corresponds to a different lower bound. These bounds provide some insight on how the auctioneer's revenue at the best pure-strategy equilibrium varies with $p$ and $n$.

\section{Symmetric Competition Against a High-Value Bidder}
\label{sec_special}

In this section, we consider a special case, where there is one bidder with a higher private value, and the other bidders have lower private values that are equal to each other. Without loss of generality, we assume that bidder $1$ has valuation $v_1=\alpha >1$, and bidders $2, \ldots, n$ have valuations $v_2=\ldots=v_n=1$. We call this case OLOS (one larger, others symmetric). The case of two bidders is a special case of OLOS. We use OLOS to gain insight about how $p$, the ratio $\alpha$ between the highest valuation and lower valuations, and the number of bidders $n$ affect revenue when a single bidder with a high private value competes against a set of bidders with lower private values. 

Having a single large bidder is particularly interesting since it yields a low equilibrium revenue in a second-price auction. Specifically, for OLOS,
the equilibrium revenue of the second-price auction is $1$, for any $\alpha>1$ and $n\geq2$. This is very unsatisfactory when $\alpha$ or $n$ is large. 
Intuitively, we expect that higher $\alpha$ should lead to higher equilibrium revenue, and more bidders (larger $n$) should spur more fierce competition, 
which should also increase the equilibrium revenue. In this section, we show by analysis and computational studies that we can adaptively design quasi-proportional auctions (i.e. adaptively choose exponent $p$ for each
$(n,\alpha)$) to achieve equilibrium revenues significantly larger than $1$, and strictly increasing in both
$\alpha$ and $n$.

As we will show in this section, the assumption of symmetric small bidders can significantly simplify the analysis.
Specifically, this assumption allows us to (1) prove the uniqueness of the pure-strategy Nash equilibrium and (2)
compute the auctioneer's equilibrium revenue efficiently. We believe that, in practice, this assumption is not as restrictive as it seems:
many practical cases in which there is one large bidder and many similar small bidders are well approximated by OLOS.

This section proceeds as follows: first, we show that for any choice of $\alpha>1$, $n\geq2$, and $p>0$,
there is a unique pure-strategy Nash equilibrium for OLOS. 
Second, we discuss how to efficiently compute the auctioneer's revenue at this unique pure-strategy Nash equilibrium,
 and provide upper and lower bounds on it.
 Finally, we discuss how to design the optimal quasi-proportional auction for each $(n, \alpha)$,
 and how $R^*(n,\alpha)$ and $p^*(n, \alpha)$ vary with $n$ and $\alpha$, where
 $R^*(n,\alpha)$ is the highest equilibrium revenue achieved by the quasi-proportional auctions
 for $(n,\alpha)$, and $p^*(n,\alpha)$ is the exponent achieving this highest equilibrium revenue.
 


\subsection{Uniqueness of the Pure-Strategy Nash Equilibrium}

We first prove that there is a unique pure-strategy Nash equilibrium for OLOS.
The following lemma shows that in such cases, a pure-strategy Nash equilibrium must have symmetric bids among 
all the small bidders.
\begin{lemma}
\label{lemma_1}
Assume that $\mathbf{b}^*$ is a pure-strategy Nash equilibrium, then we have
$$b_2^*=b_3^*=\ldots = b_n^*.$$
\end{lemma}
(Please refer to the appendix for the proof.)
Thus, to characterize $\mathbf{b}^*$, we only need to specify $b_1^*$ and $b_2^*$.
To simplify the exposition, in the remainder of this section, we sometimes use $b_1$ and $b_2$ to respectively denote $b_1^*$ and $b_2^*$, and define
$z=b_1/b_2$.
Note that from the first-order condition for the best response, we have
\begin{eqnarray}
(n-1) b_2^p \left[ \alpha p - (1+p) b_1 \right] &=& b_1^{p+1} \nonumber \\
\left[ b_1^p + (n-2) b_2^p\right] \left[ p- (1+p) b_2 \right] &=& b_2^{p+1}  \label{eqn:case_n:1}
\end{eqnarray}
Notice that the above equations imply that
$b_1 < \frac{p}{1+p} \alpha$ and $b_2 < \frac{p}{1+p}$.
Dividing the first equation by the second equation\footnote{It is straightforward to prove that at a Nash equilibrium, $b_1 >0$ and $b_2>0$, so the division is well-defined.}, we have
\begin{equation}
\frac{n-1}{z^p+n-2} \frac{\alpha p -(1+p) z b_2}{p - (1+p) b_2} =z^{p+1}. 
\end{equation}
Solving the above equation, we have
\begin{equation}
b_2=\frac{p}{1+p} \frac{\alpha-w}{z-w}, \quad \text{where} \quad w=z^{p+1} \frac{z^p+n-2}{n-1}. \label{eqn:case_n:2}
\end{equation}
Substituting the above equation into the first equation of Equation \ref{eqn:case_n:1}, we have
\[
(n-1) \left[ \alpha - \frac{w-\alpha}{w-z} z \right] = \frac{z^{p+1}}{1+p} \frac{w-\alpha}{w-z}, 
\]
which implies that
\[
w = \frac{\alpha z^{p+1}}{z^{p+1} + (1+p)(n-1)(z-\alpha)}.
\]
Combining the above equation with Equation \ref{eqn:case_n:2}, we have
\[
\frac{z^p+n-2}{n-1} = \frac{\alpha }{z^{p+1} + (1+p)(n-1)(z-\alpha)},
\]
which is
\begin{eqnarray}
h(z) &\stackrel{\Delta}{=}&
z^{2p+1} + \left[ (n-2) + (1+p)(n-1)\right] z^{p+1} - \alpha (1+p)(n-1) z^p  \nonumber \\
&+& (n-1)(n-2)(1+p)z - \alpha (n-1) \left[ (1+p)(n-2)+1\right] =0, \label{eqn:h_equation}
\end{eqnarray}
where $h(z)$ is a shorthand notation for the lefthand side of the above equation.
The following lemma states that Equation \ref{eqn:h_equation} has no root in
$\left[0,  \alpha^{\frac{1}{1+2p}}\right] \bigcup \left[ \alpha, \infty \right)$.
\begin{lemma}
\label{lemma:case_n_2}
$h(z)>0$ for all $z \geq \alpha$, and $h(z)<0$ for all $z \in \left[0, \alpha^{\frac{1}{1+2p}} \right]$.
\end{lemma}
Please refer to the appendix for the proof of Lemma \ref{lemma:case_n_2}. The following lemma states that
Equation \ref{eqn:h_equation}  has a unique solution in $\Re^+$, 
whose proof is also available in the appendix.
\begin{lemma}
\label{lemma:case_n_3}
The equation $h(z)=0$ has a unique solution $\tilde{z} \in \Re^+$. Moreover, $\tilde{z} \in \left ( \alpha^{\frac{1}{2p+1}}, \alpha \right )$. 
\end{lemma}
Putting together all the results in this section, we have the following theorem:
\begin{theorem}
\label{thm_uniqueness}
For OLOS, there is a unique pure-strategy Nash equilibrium $\mathbf{b}^*$. 
Moreover:
\begin{enumerate}
\item This equilibrium is symmetric among all the small bidders
in the sense that $b_2^* = b_3^* =\ldots = b_n^*$.
\item At this equilibrium, $\alpha^{\frac{1}{1+2p}} < b_1^*/ b_2^* < \alpha$ and
$
\frac{p}{1+p + \frac{1}{n-1}}\leq b_2^* < \frac{p}{1+p}$.
\end{enumerate}
\end{theorem}
\begin{proof}
The symmetry among the small bidders follows from Lemma~\ref{lemma_1}.
From Lemma~\ref{lemma:case_n_3}, equation $h(z)=0$ has a unique solution in $\left ( \alpha^{\frac{1}{2p+1}}, \alpha \right )$. 
Hence, $z=b_1^*/b_2^*$ is uniquely determined and $\alpha^{\frac{1}{1+2p}} < b_1^*/ b_2^* < \alpha$. From Equation \ref{eqn:case_n:2}, 
once $z$ is uniquely determined, $b_2^*$ and $b_1^*= z b_2^*$ are also uniquely determined.
Thus, there exists a unique pure-strategy Nash equilibrium $b^*$.
Finally, $b_2^* < \frac{p}{1+p}$ follows from Equation \ref{eqn:case_n:1}, and 
$b_2^* \geq \frac{p}{1+p + \frac{1}{n-1}} $ follows from Corollary \ref{corollary_1}.
\end{proof}

\subsection{Auctioneer's Equilibrium Revenue}

We now characterize the auctioneer's equilibrium revenue. Notice that
\begin{eqnarray}
R &=& \frac{b_1^p}{b_1^p + (n-1) b_2^p} b_1 + \frac{(n-1)b_2^p}{b_1^p + (n-1) b_2^p} b_2 \nonumber \\
&=&  \frac{z^p}{z^p + (n-1) } z b_2 + \frac{(n-1)}{z^p + (n-1)} b_2 \nonumber \\
&=& \frac{z^{p+1} +n-1}{z^p +n-1} b_2. \nonumber
\end{eqnarray}
From Equation \ref{eqn:case_n:2}, we have
\begin{equation}
R=\frac{p}{1+p} \left[ 1+ \frac{z^p (z-1)}{z^p+ (n-1)}\right] \left[ 1- \frac{(\alpha-z)(n-1)}{z^{2p+1} + (n-2) z^{p+1} -z(n-1)}\right], \label{eqn_R}
\end{equation}
where $z$ is the unique solution of Equation \ref{eqn:h_equation}. 
Thus, we can efficiently compute the auctioneer's equilibrium revenue as follows:
\begin{enumerate}
\item Compute $z$ by (numerically) solving
Equation \ref{eqn:h_equation}.
\item Compute the auctioneer's equilibrium revenue $R$ by Equation \ref{eqn_R}.
\end{enumerate}
Since OLOS and the quasi-proportional mechanism
are fully characterized by the triple $(n, \alpha, p)$, sometimes we write 
$R$ as $R(n, \alpha, p)$ to emphasize this dependence. 
It is straightforward to derive the following bounds
on $R(n, \alpha, p)$:
\begin{corollary}
\label{corollary_3}
For any $(n, \alpha, p)$, we have $R(n, \alpha, p) < \frac{p}{1+p} \left[ 1+ \frac{\alpha^p (\alpha-1)}{\alpha^p + (n-1)}\right]$ and
\small
\[
R(n, \alpha, p) > \frac{p}{1+p} \left[ 1+ \frac{\alpha^{\frac{p}{2p+1}} \left( \alpha^{\frac{1}{2p+1}}-1 \right)}{\alpha^{\frac{p}{2p+1}} + (n-1)}\right] \left[ 1- \frac{(\alpha-\alpha^{\frac{1}{2p+1}})(n-1)}{\alpha + (n-2) \alpha^{\frac{p+1}{2p+1}} -(n-1) \alpha^{\frac{1}{2p+1}}}\right].
\]
\normalsize
\end{corollary}
\begin{proof}
Define
\[
\eta(x)=\frac{p}{1+p} \left[ 1+ \frac{x^p (x-1)}{x^p+ (n-1)}\right] \left[ 1- \frac{(\alpha-x)(n-1)}{x^{2p+1} + (n-2) x^{p+1} -x(n-1)}\right],
\]
it is straightforward to show that $\eta(x)$ is strictly increasing in $x$ over interval $\left( \alpha^{\frac{1}{2p+1}}, \alpha \right)$.
This corollary follows from the facts that (1) $R(n, \alpha, p) =\eta(z)$, and (2) $z \in \left( \alpha^{\frac{1}{2p+1}}, \alpha \right)$.
\end{proof}

\subsection{Revenue-Maximizing Mechanism Design}

Now we discuss how to design a quasi-proportional auction mechanism (i.e. how to choose the exponent $p$)
to maximize the auctioneer's equilibrium revenue for OLOS.
Notice that for a given $(n, \alpha)$, this mechanism design problem is a one-dimensional optimization problem.
Moreover, as we have discussed above, for a given triple $(n, \alpha, p)$, the auctioneer's equilibrium revenue $R(n, \alpha, p)$ can be
efficiently computed. Consequently, this mechanism design problem can be efficiently solved by line search. We present some computational analysis of $R(n, \alpha, p)$, followed by some theorems about how $n$, $\alpha$, and $p$ influence equilibrium revenue. 

In Figure~\ref{revenue_plot}, we plot $R(n, \alpha, p)$ versus $p$ for some given
$(n, \alpha)$'s. 
Notice that the value of $p$ that maximizes revenue increases as $\alpha$ decreases and as $n$ increases. In other words, as competition increases, steeper weight functions are needed to maximize revenue. 
Also notice that for all the considered $(n, \alpha)$'s, there exists $p$ such that
$R(n, \alpha, p)>1$, the equilibrium revenue of the second-price auction.
Moreover, for all the considered $(n, \alpha)$'s, the equilibrium revenue $R$
first increases with $p$, and then decreases with $p$. So there is a unique 
revenue-maximizing $p$ (quasi-proportional mechanism).

In Figure~\ref{revenue_plot_alpha}, we plot $R(n, \alpha, p)$ versus $\alpha$ for some given $(n,p)$'s.
We observe that $R$ is an increasing function of $\alpha$ for all the considered  $(n,p)$'s, which is reasonable since
a higher private value for the large bidder will lead to a higher equilibrium revenue, for a fixed number of small bidders and value $p$.

In Figure~\ref{revenue_plot_n}, we plot $R(n, \alpha, p)$ versus $n$ for some $(\alpha,p)$ values.
Notice that Figures~\ref{decrease_n_p3_small} and \ref{decrease_n_p5_small} 
show that for a fixed $(\alpha,p)$, the equilibrium revenue $R$ first increases with $n$ and then decreases with $n$.
As more small bidders are added, their equilibrium bids increase due to increased competition. 
For fixed $\alpha$ and $p$, this can have two conflicting effects: 
initially, the increased bids increase the equilibrium revenue, but as more small bidders are added, the small bidders win more of the allocation, reducing the portion bought by the single high-value bidder at a higher bid. This decreases revenue.
This revenue reduction can be alleviated by adapting $p$ to
$(n,\alpha)$, that is, choosing the revenue-maximizing $p$ 
based on $(n,\alpha)$. We will discuss this later in this section.
We also notice that the initial interval where $R$ increases as a function of $n$ depends on $(\alpha,p)$.
Figures~\ref{increase_n_p3_large} and \ref{increase_n_alpha3} show that if  $\alpha$ or $p$ is large,
$R$ is monotonically increasing for $n\leq 600$.

We are interested in the highest equilibrium revenue achieved by the quasi-proportional mechanisms, and which
quasi-proportional mechanism achieves the highest revenue. Thus, we define
\be
\label{eqn_star}
R^*(n, \alpha)=\max_{p>0} R(n, \alpha, p) \quad \text{ and } \quad p^*(n, \alpha)=\argmax_{p>0}  R(n, \alpha, p) .
\ee

Figure~\ref{R_star_alpha} shows how $R^*(n, \alpha)$ and $p^*(n, \alpha)$ vary with $\alpha$ and $n$. Specifically, Figures~\ref{vary_alpha_R} and \ref{vary_alpha_R_1} show that for a given
$n$, $R^*$ is a strictly increasing function of $\alpha$. 
That is, the auctioneer's highest equilibrium revenue increases with
the large bidder's private value.
Moreover, it is interesting to observe that $R^*$ is almost an affine function of $\alpha$, 
and different $n$'s corresponds to different slopes.
Figures~\ref{vary_alpha_p} and \ref{vary_alpha_p_1} show that for a given $n$,
$p^*$ is a strictly decreasing function of $\alpha$.
Moreover, the decrease rate (the negative of the first derivative) decreases as
$\alpha$ increases. 
Note that $p^*>1$ when $\alpha$ is small, which implies that the revenue-maximizing weight function is convex.
We also observe that $p^*$ is primarily determined by $\alpha$: 
for a fixed $\alpha$ and very different $n$'s (e.g. $2$ and $500$), the variation of $p^*$ is small.


Figure~\ref{vary_n_R} shows that for a given $\alpha$,
$R^*$ is a strictly increasing function of $n$. That is, more small bidders will increase the auctioneer's equilibrium revenue.
We also observe that as $n$ increases, $R^*(n+1, \alpha)- R^*(n, \alpha)$ decreases.
That is, each additional small bidder makes a smaller marginal contribution to the auctioneer's highest equilibrium revenue.
Figure~\ref{vary_n_p} shows that for a given $\alpha$, $p^*$ first decreases with $n$ and then increases with $n$.
Similar to Figures~\ref{vary_alpha_p} and \ref{vary_alpha_p_1}, Figure~\ref{vary_n_p} demonstrates that for a given $\alpha$,
the variation of $p^*$ with $n$ is relatively small, which implies the robustness of $p^*$ (mechanism design) to the mis-specification of parameter
$n$. 

\begin{figure}[t]
\subfigure[$\alpha=3$]{\includegraphics[width = 0.48\textwidth]{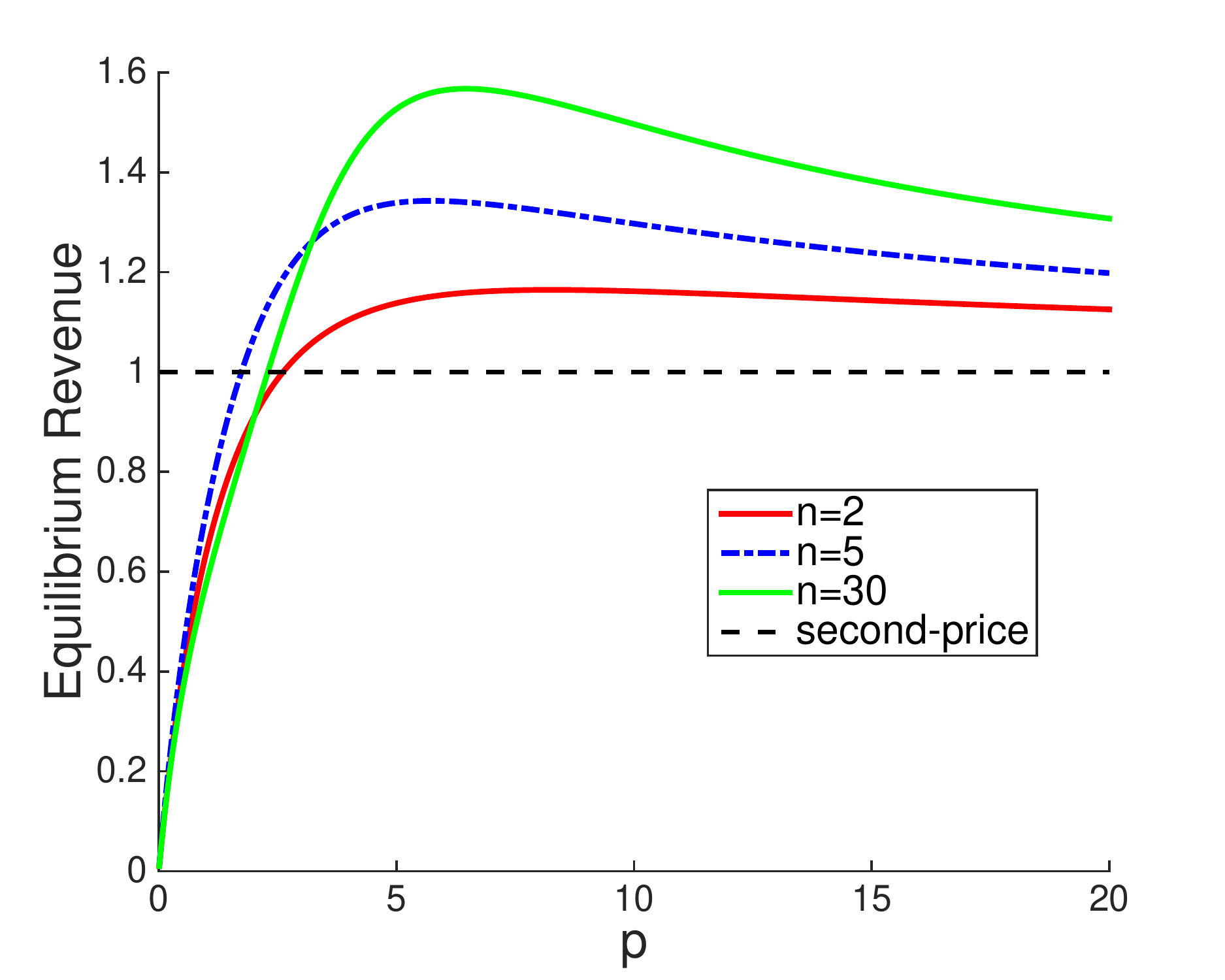} \label{revenue_plot_1}}
\subfigure[$\alpha=10$]{\includegraphics[width = 0.48\textwidth]{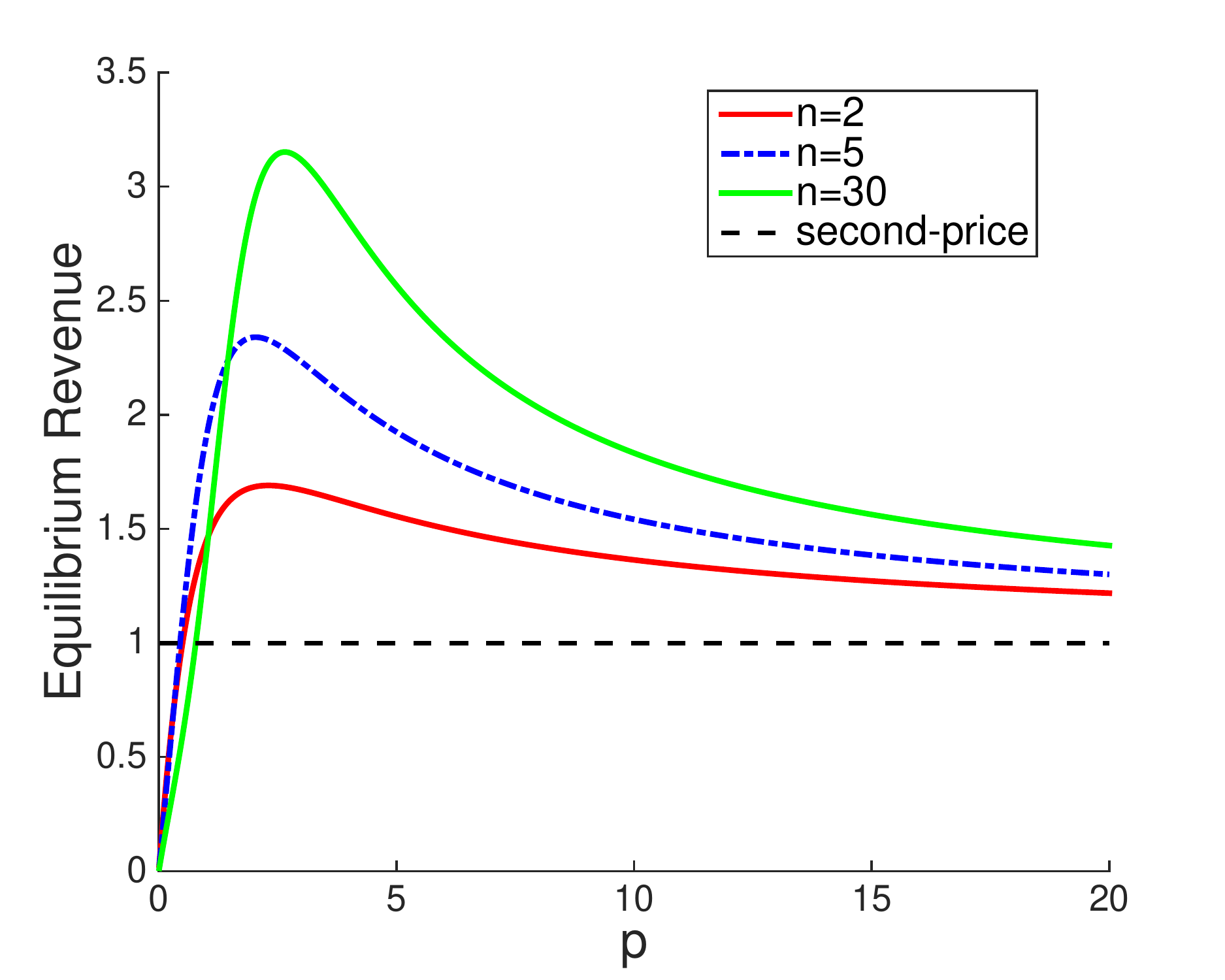} \label{revenue_plot_2}} \\
\subfigure[$n=2$]{\includegraphics[width = 0.48\textwidth]{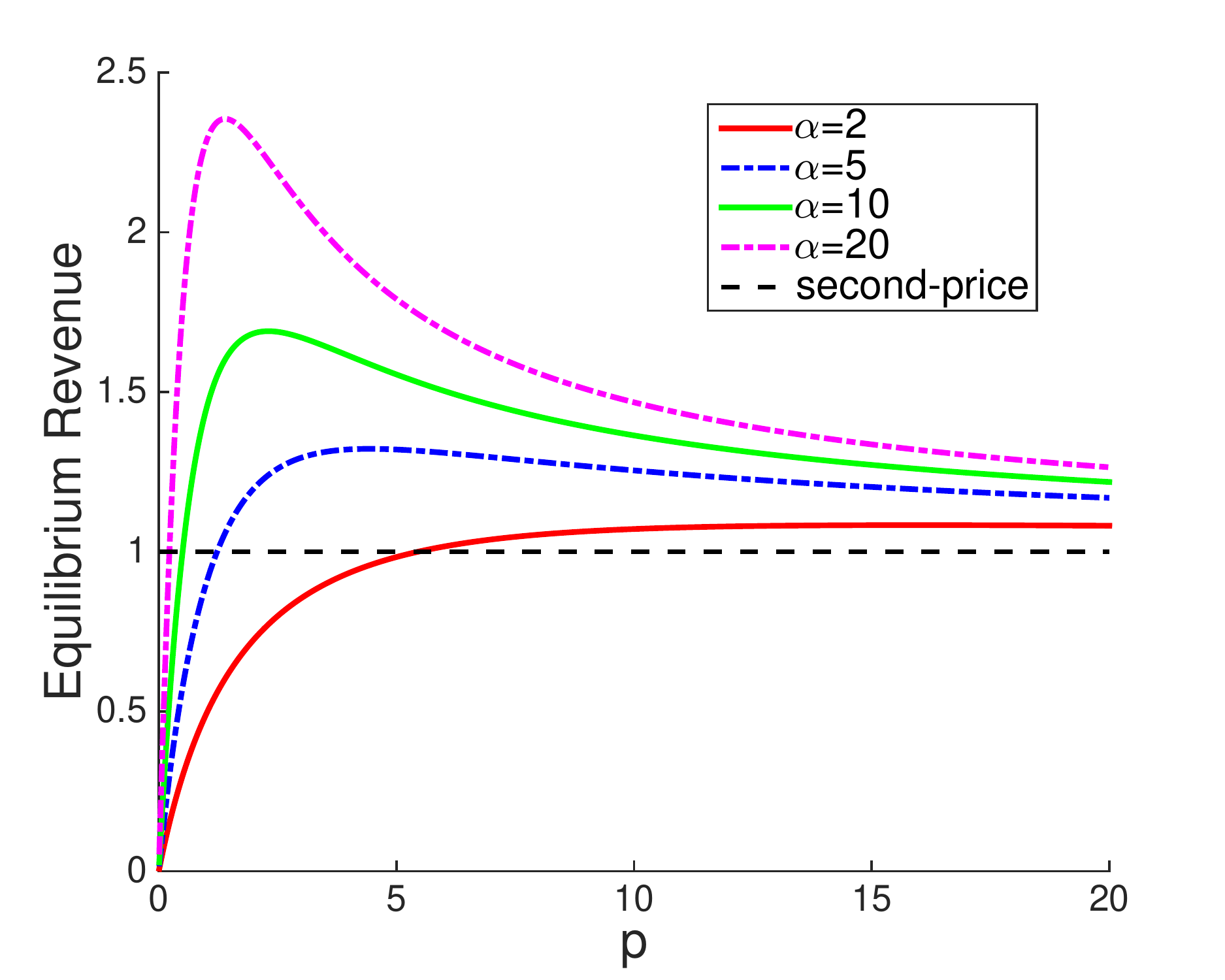} \label{revenue_plot_3}}
\subfigure[$n=20$]{\includegraphics[width = 0.48\textwidth]{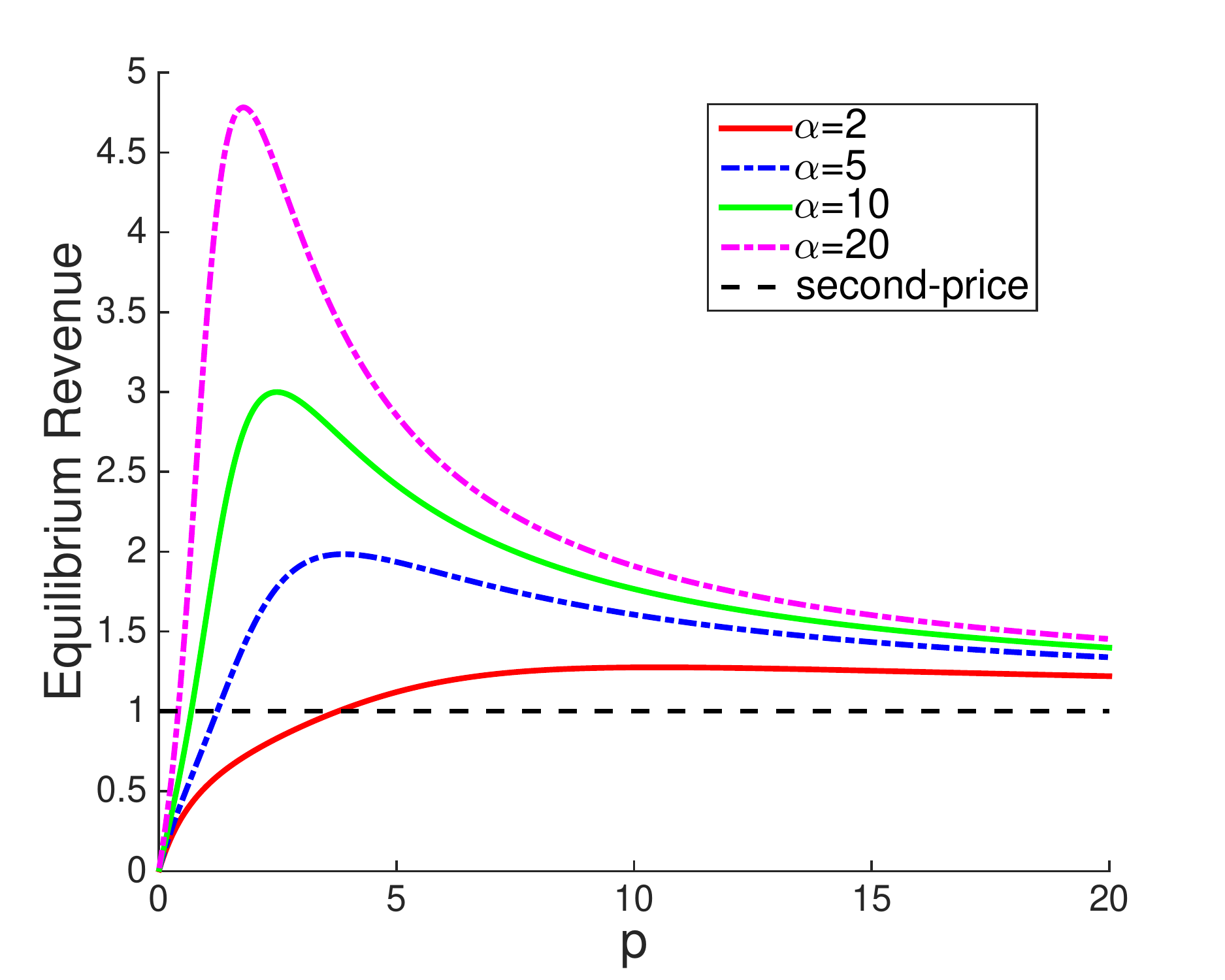} \label{revenue_plot_4}}
\caption{Equilibrium Revenue $R(n,\alpha,p)$ vs. $p$}
\label{revenue_plot}
\end{figure}

\begin{figure}[ht]
\subfigure[$p=3$]{\includegraphics[width = 0.48\textwidth]{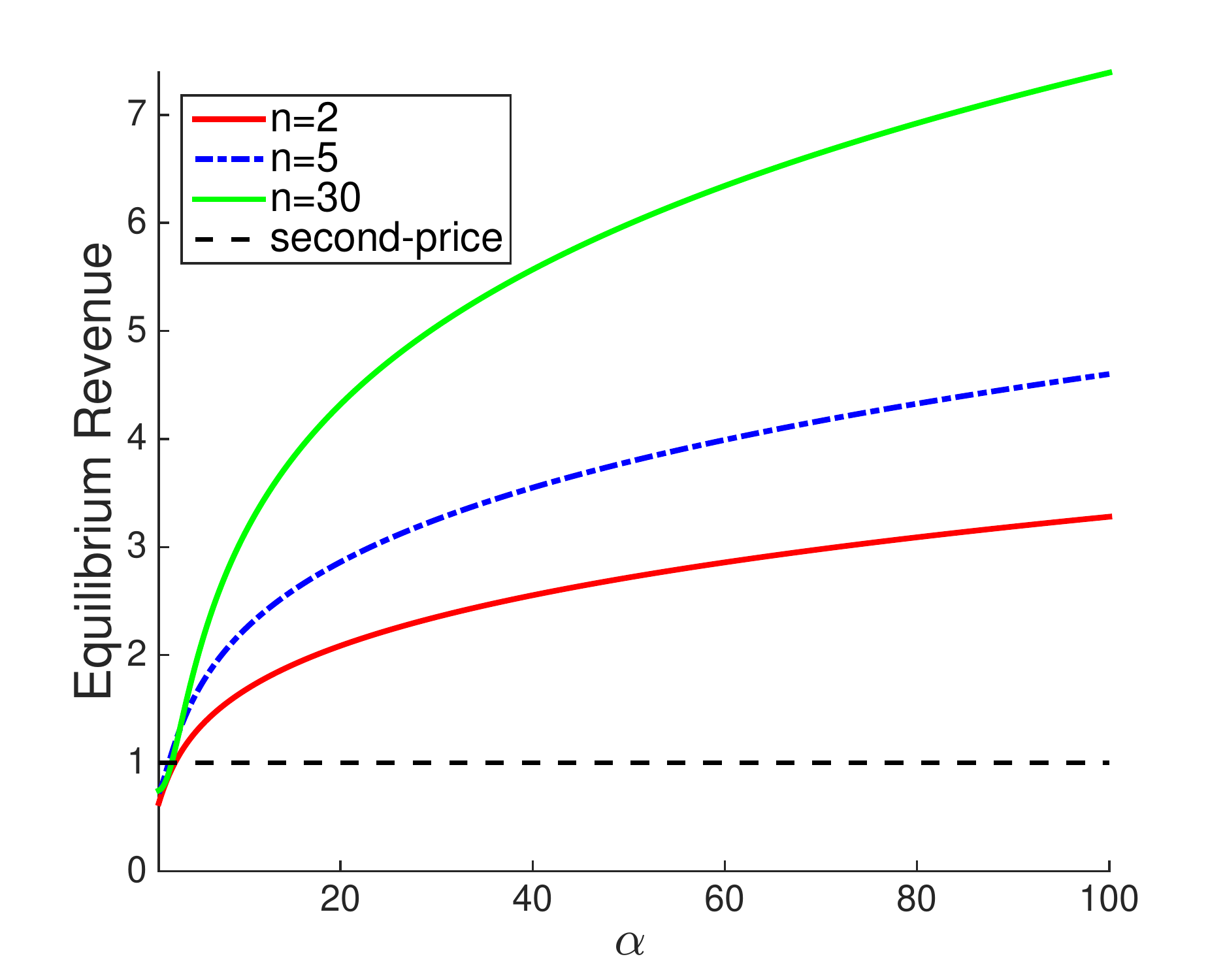} \label{increase_alpha_p3}}
\subfigure[$p=5$]{\includegraphics[width = 0.48\textwidth]{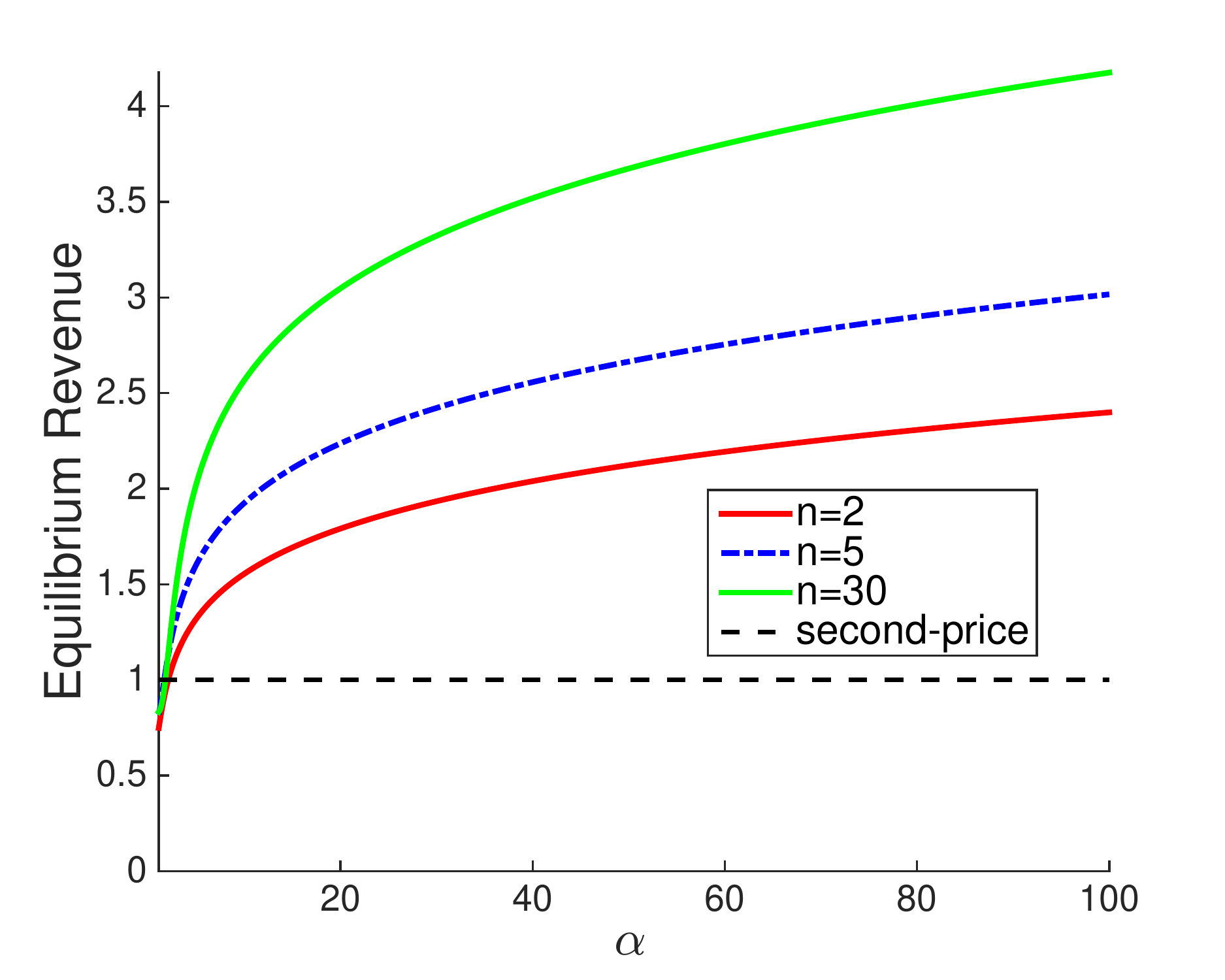} \label{increase_alpha_p5}}
\caption{Equilibrium Revenue $R(n,\alpha,p)$ vs. $\alpha$}
\label{revenue_plot_alpha}
\end{figure}

\begin{figure}[ht]
\subfigure[$p=3$, small $\alpha$]{\includegraphics[width = 0.48\textwidth]{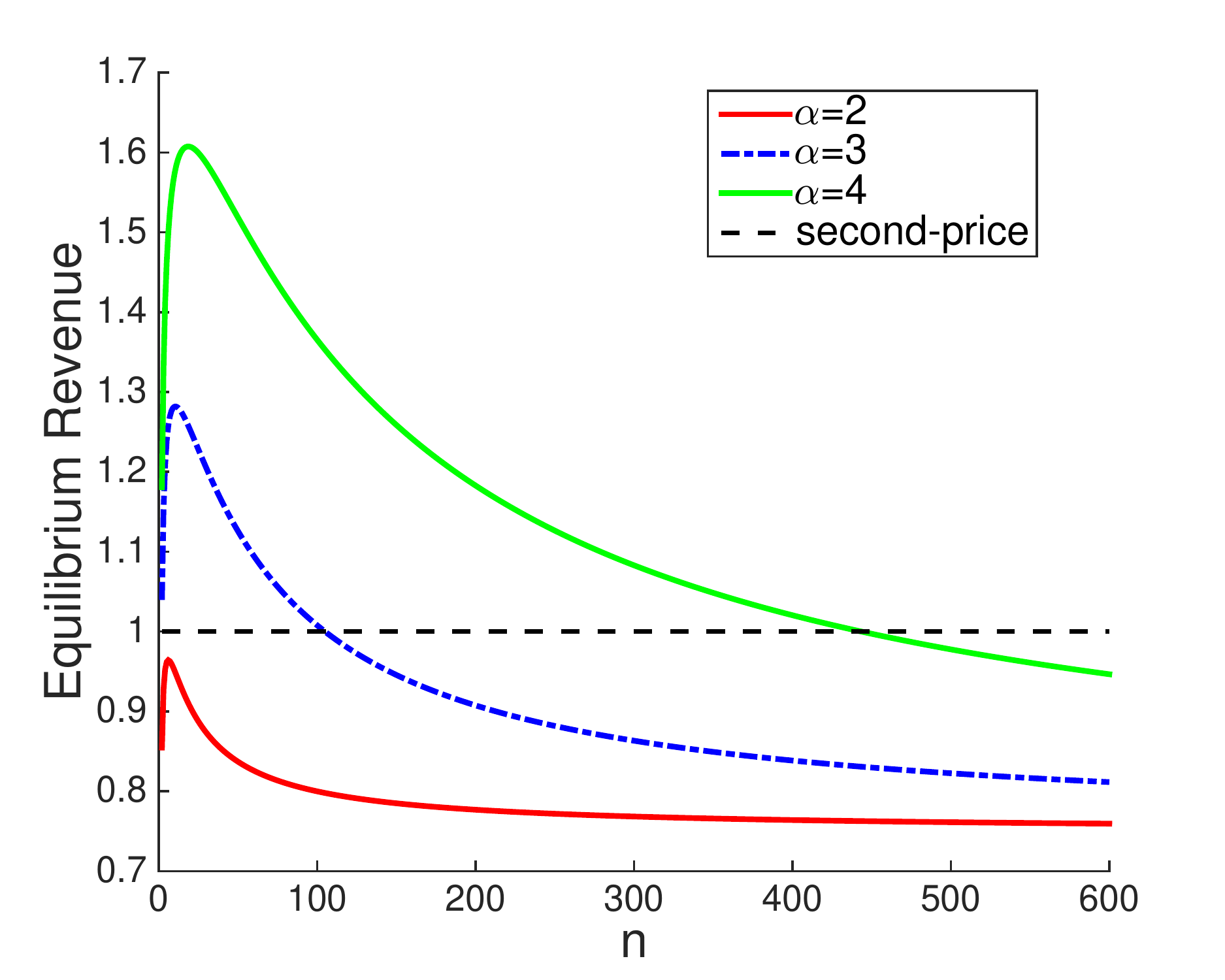} \label{decrease_n_p3_small}}
\subfigure[$p=5$, small $\alpha$]{\includegraphics[width = 0.48\textwidth]{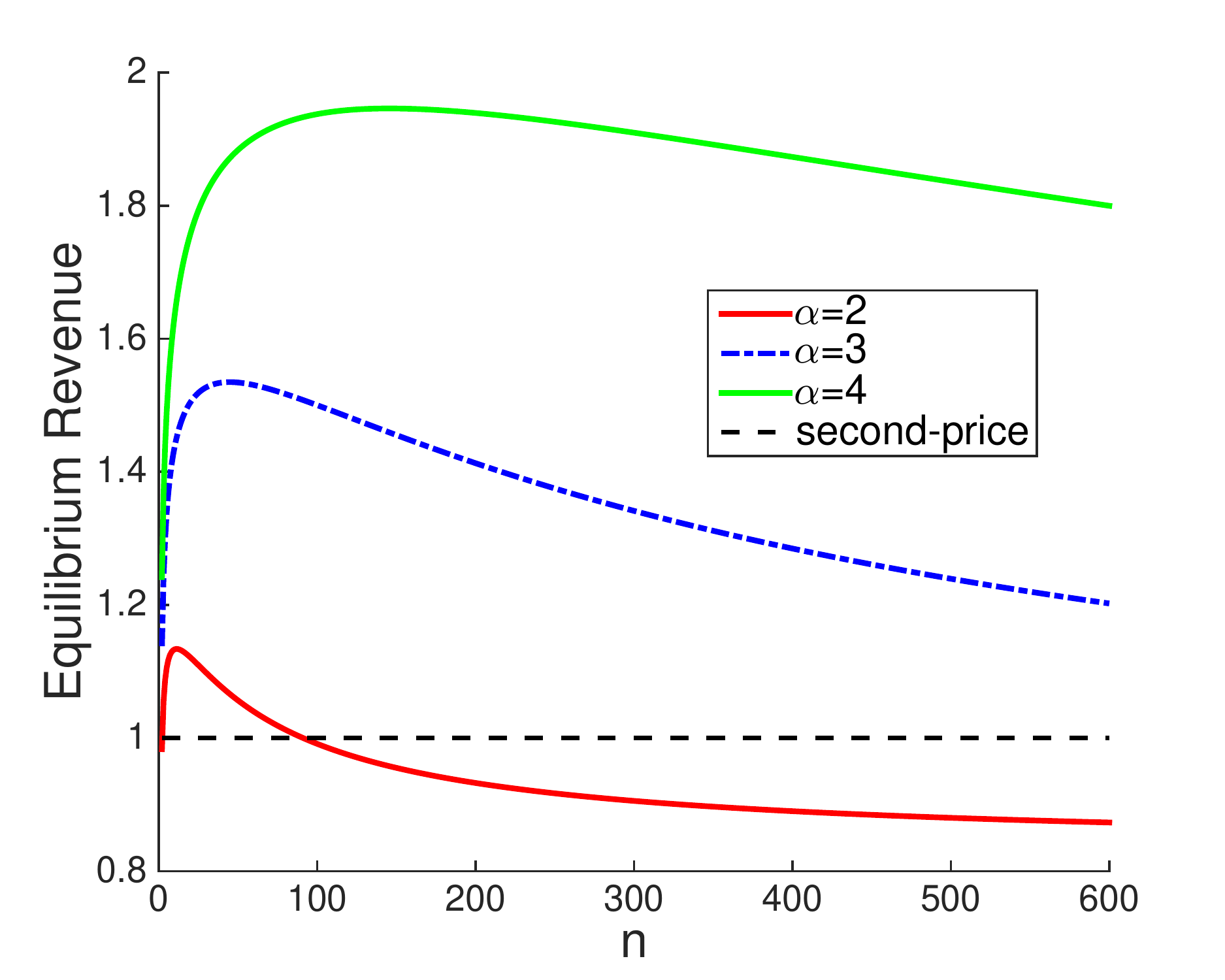} \label{decrease_n_p5_small}}\\
\subfigure[$p=3$, large $\alpha$]{\includegraphics[width = 0.48\textwidth]{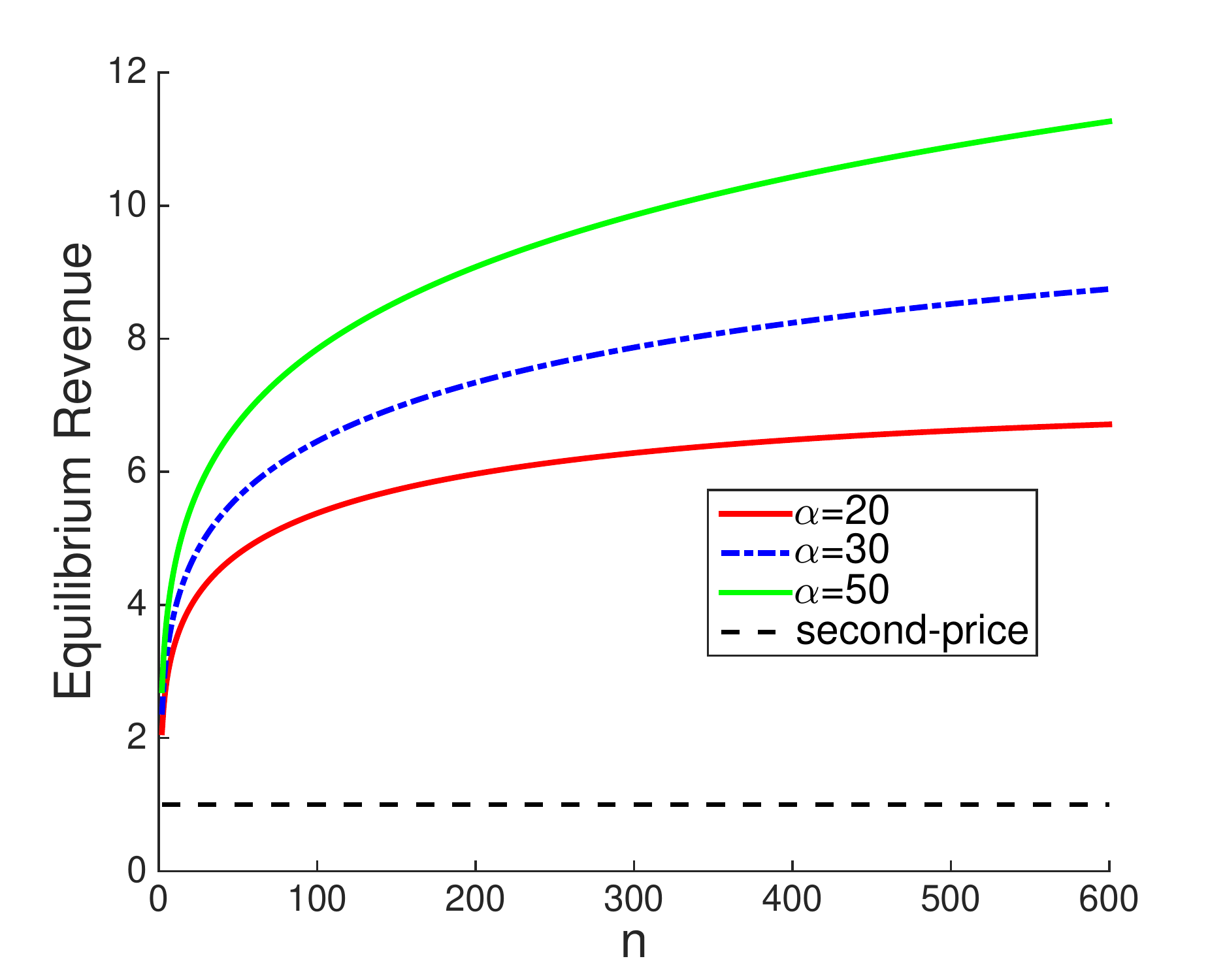} \label{increase_n_p3_large}}
\subfigure[$\alpha=3$, large $p$]{\includegraphics[width = 0.48\textwidth]{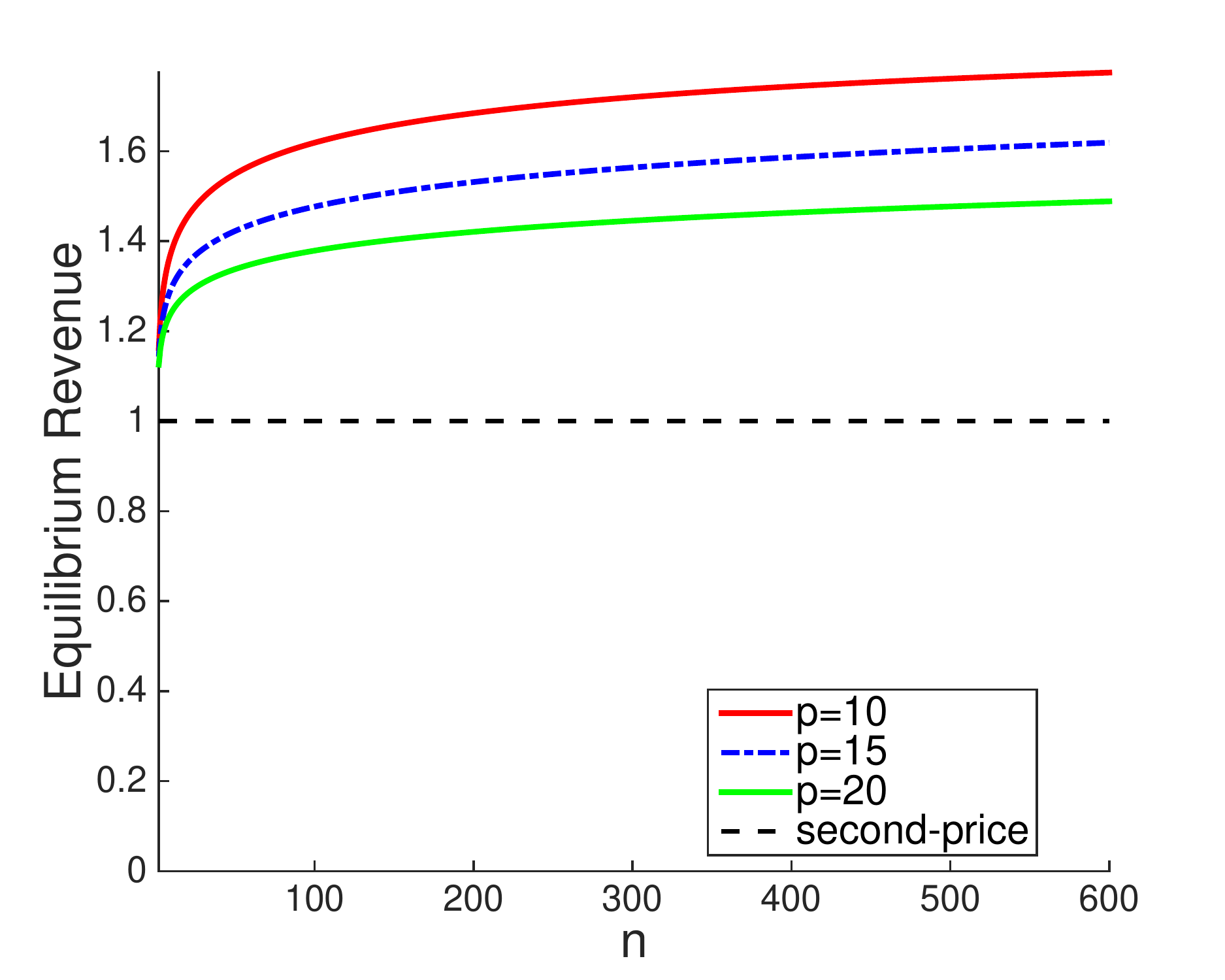} \label{increase_n_alpha3}}
\caption{Equilibrium Revenue $R(n,\alpha,p)$ vs. $n$}
\label{revenue_plot_n}
\end{figure}

\begin{figure}[th]
\subfigure[$R^*$ vs. $\alpha$]{\includegraphics[width = 0.48\textwidth]{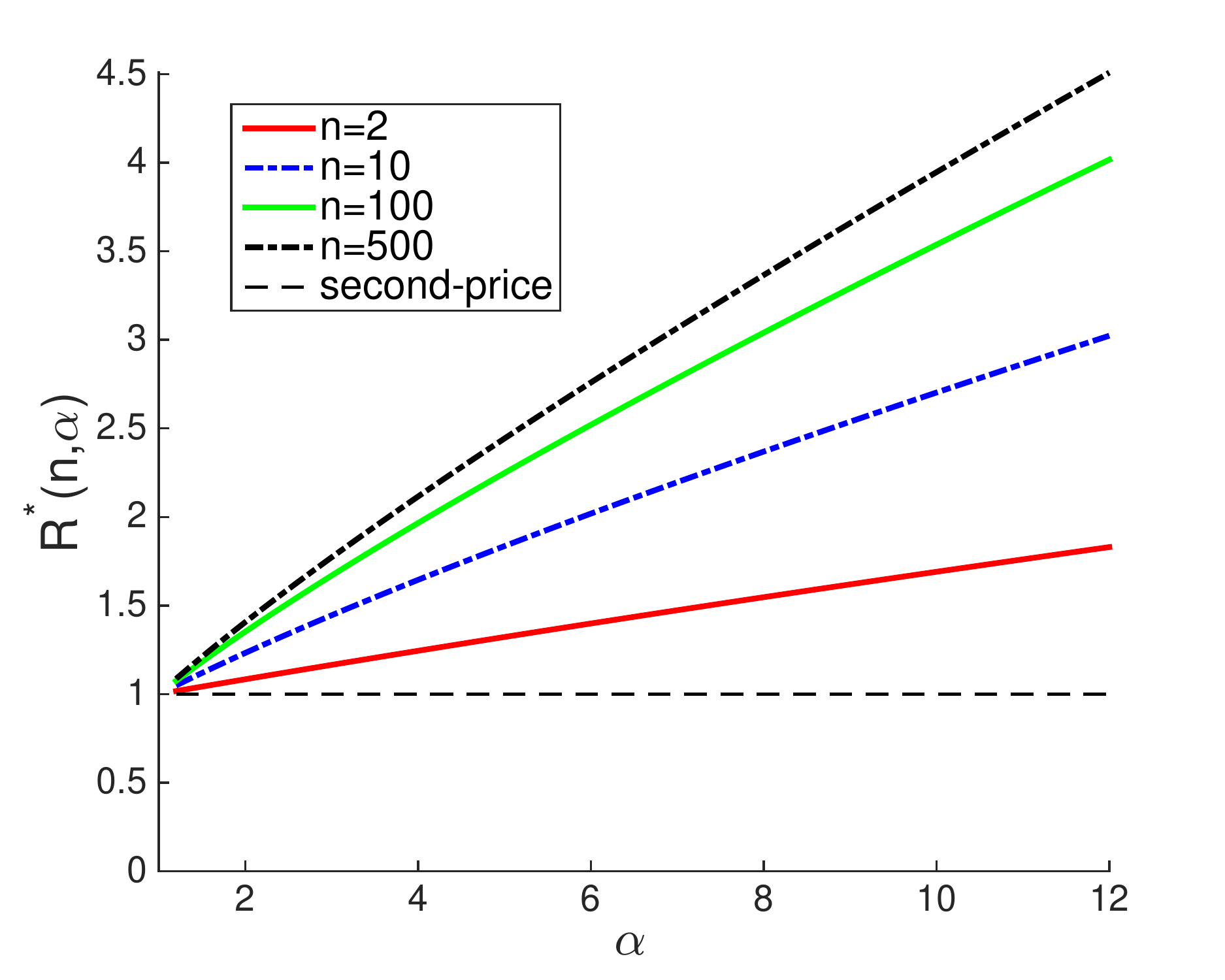} \label{vary_alpha_R}}
\subfigure[$p^*$ vs. $\alpha$]{\includegraphics[width = 0.48\textwidth]{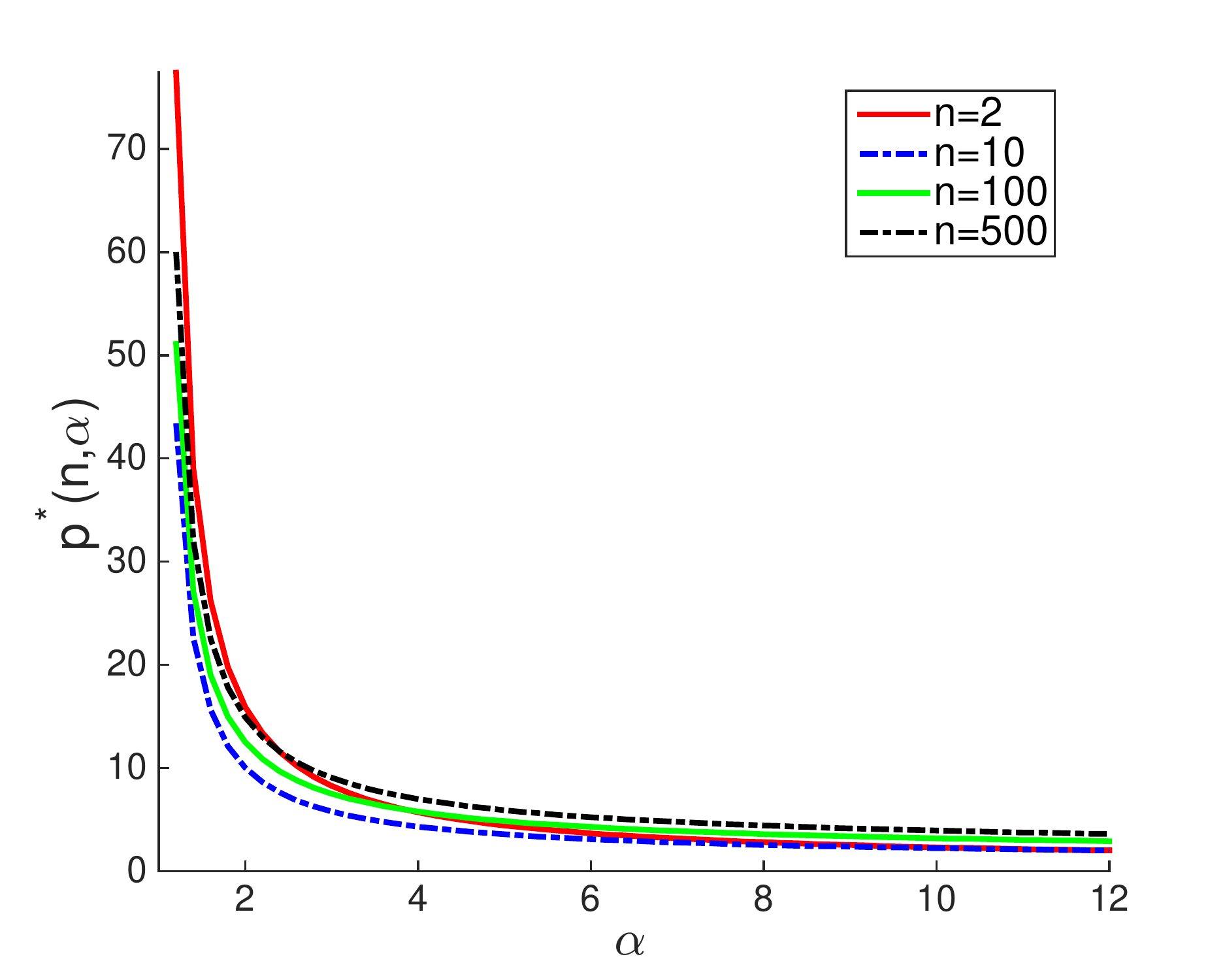} \label{vary_alpha_p}}\\
\subfigure[$R^*$ vs. $\alpha$]{\includegraphics[width = 0.48\textwidth]{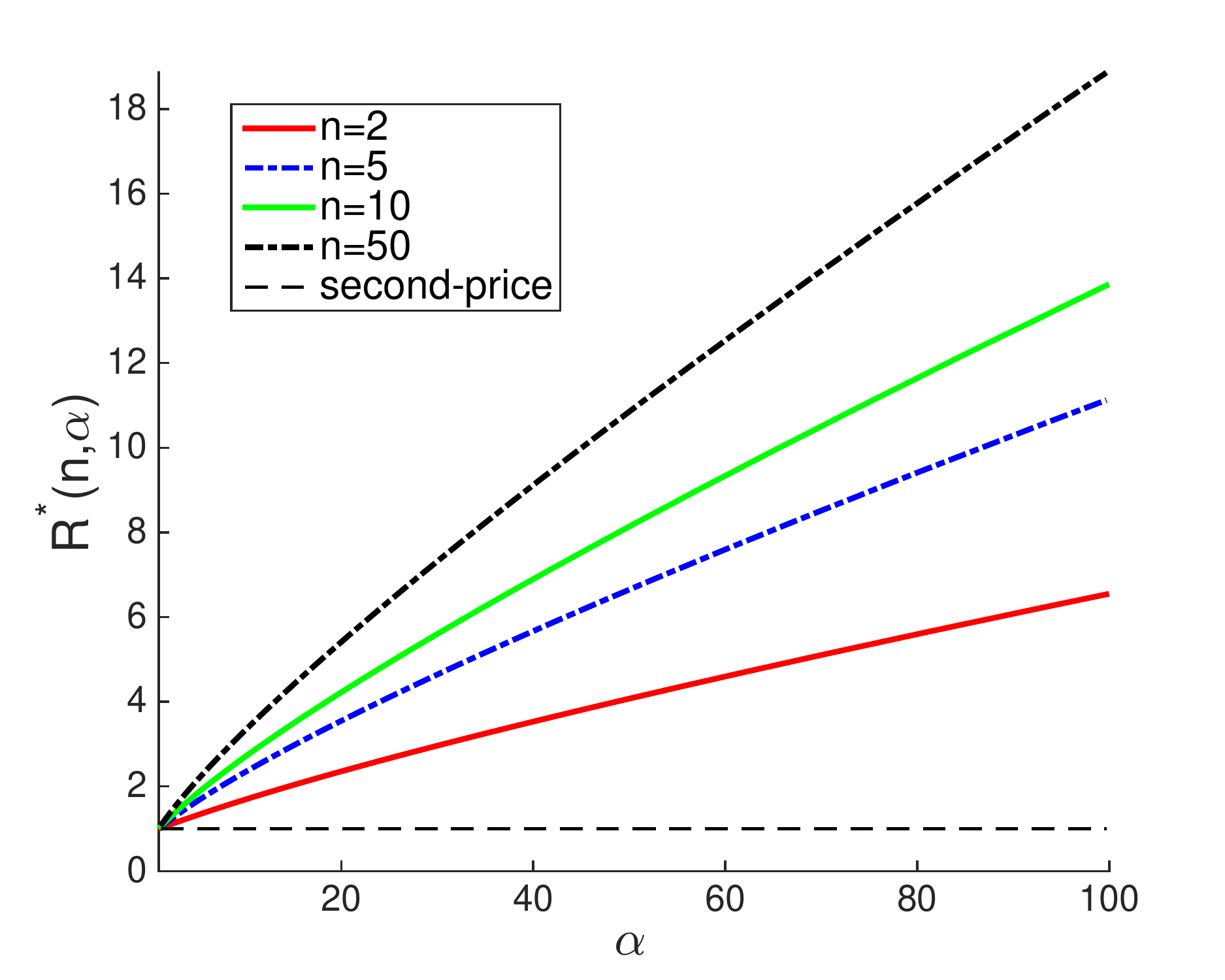} \label{vary_alpha_R_1}}
\subfigure[$p^*$ vs. $\alpha$]{\includegraphics[width = 0.48\textwidth]{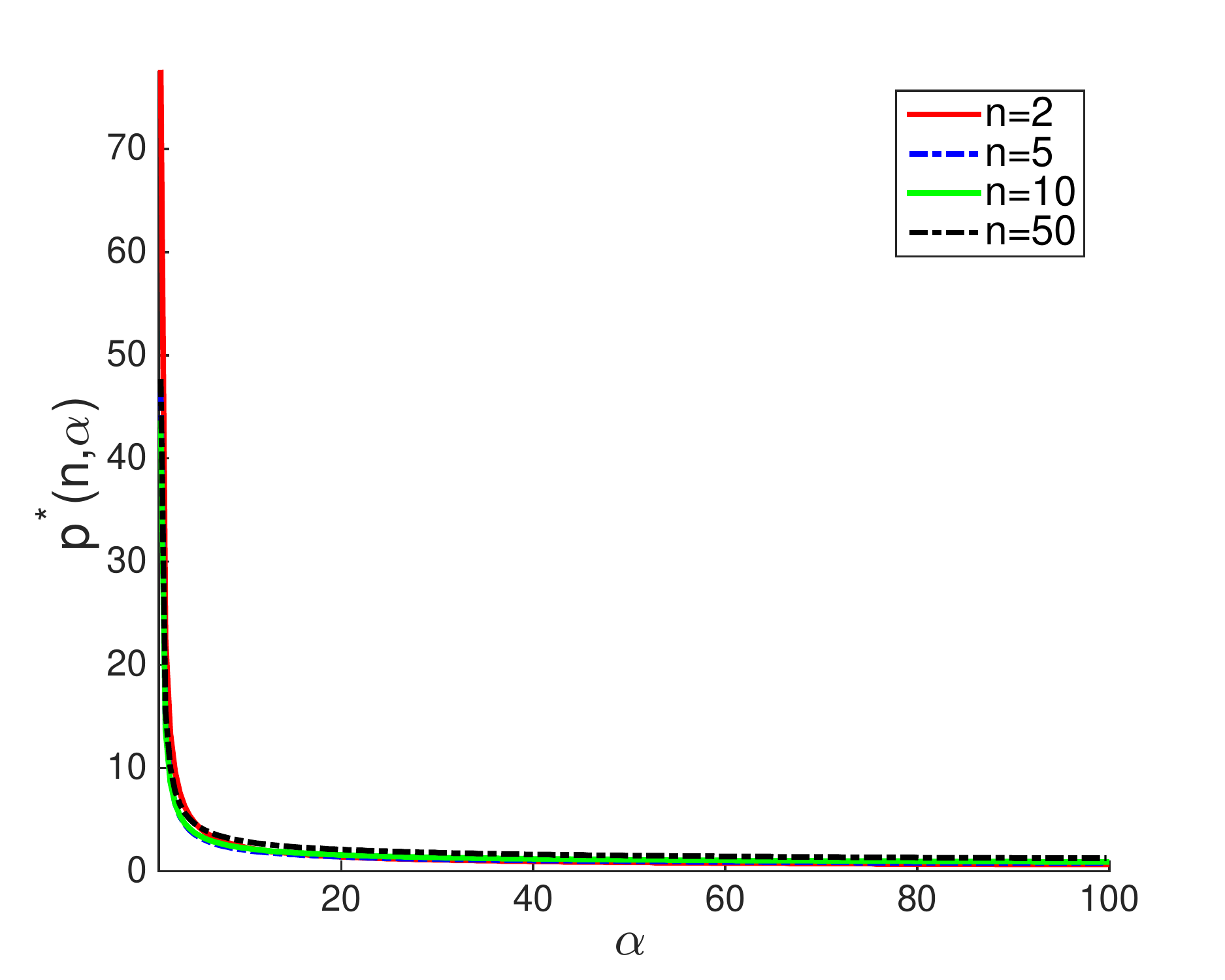} \label{vary_alpha_p_1}}\\
\subfigure[$R^*$ vs. $n$]{\includegraphics[width = 0.48\textwidth]{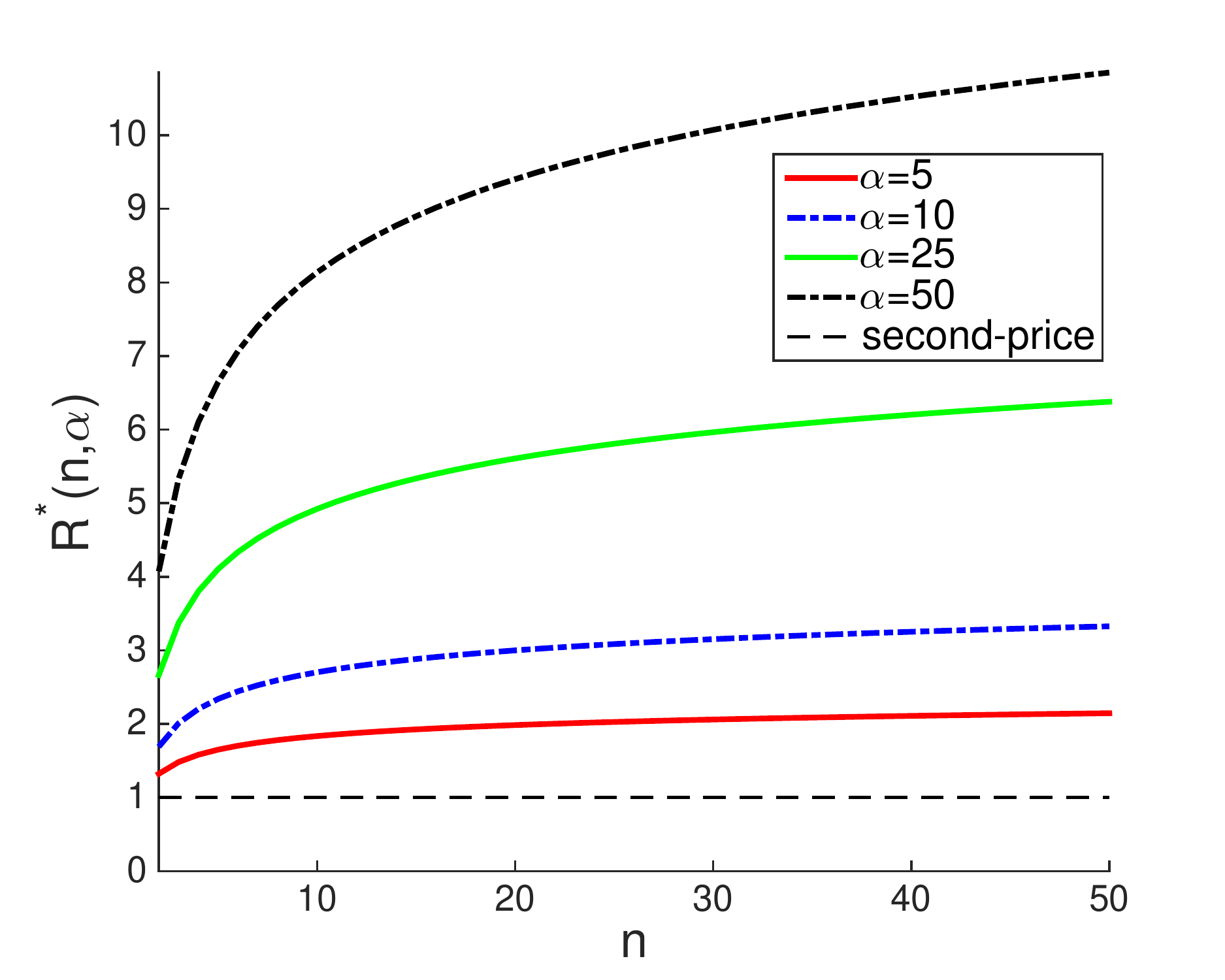} \label{vary_n_R}}
\subfigure[$p^*$ vs. $n$]{\includegraphics[width = 0.48\textwidth]{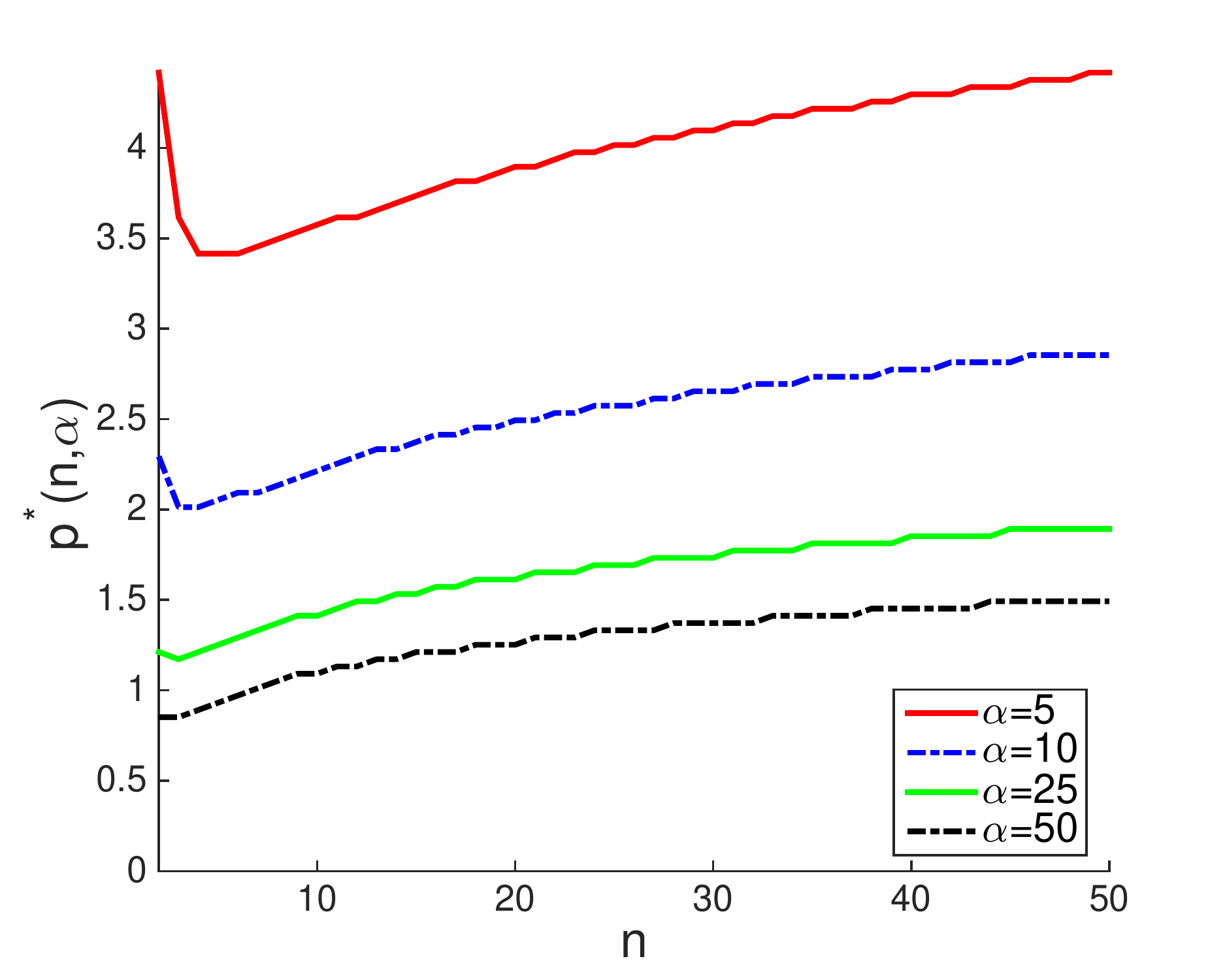} \label{vary_n_p}}
\caption{$R^* (n, \alpha)$ and $p^*(n, \alpha)$ vs. $\alpha$ or $n$}
\label{R_star_alpha}
\end{figure}

Recall that for OLOS, the equilibrium revenue from a second-price auction is always $1$.
The following theorem states that for any $(n, \alpha)$, the equilibrium revenue will diminish as $p \downarrow 0$, 
while the equilibrium revenue will be at least as much as that of the second-price auction as $p \rightarrow \infty$.
\begin{theorem}
For any $n\geq2$ and $\alpha>1$, we have
\[
\lim_{p \downarrow 0} R(n,\alpha, p)=0 \text{ and } \liminf_{p \rightarrow \infty}  R(n,\alpha, p) \geq 1.
\]
\end{theorem}
\begin{proof}
Notice that $R(n, \alpha, p) < \frac{p}{1+p} \left[ 1+ \frac{\alpha^p (\alpha-1)}{\alpha^p + (n-1)}\right]$ from Corollary \ref{corollary_3}, thus
\[
\lim_{p \downarrow 0} R(n, \alpha, p)  \leq \lim_{p \downarrow 0}  \frac{p}{1+p} \left[ 1+ \frac{\alpha^p (\alpha-1)}{\alpha^p + (n-1)}\right] =0.
\]
Hence, we have $\lim_{p \downarrow 0} R(n, \alpha, p) =0$.
Moreover, from Theorem \ref{thm_uniqueness}, 
$
\frac{p}{1+p + \frac{1}{n-1}}\leq b_2^* < \frac{p}{1+p}$. Thus, $\lim_{p \rightarrow \infty} b_2^*=1$. 
Moreover, since $b_1^*> \alpha^{\frac{1}{2p+1}} b_2^*> b_2^*$, we have
$\liminf_{p \rightarrow \infty} b_1^* \geq \lim_{p \rightarrow \infty} b_2^*=1$.
Hence we have $\liminf_{p \rightarrow \infty}  R(n,\alpha, p) \geq 1$.
\end{proof}

The following theorem shows that there does not exist a $p$ (i.e. a quasi-proportional mechanism with $f(x)=x^p$) 
such that for all $(n, \alpha)$,
$R (n, \alpha, p)>1$ (i.e. has an equilibrium revenue higher than the second-price auction). 
In other words, to achieve an equilibrium revenue higher than the second-price auction, we have to choose different quasi-proportional auction mechanisms
for different $(n, \alpha)$'s.
\begin{theorem}
For any $(n, p)$
\[
\displaystyle \limsup_{\alpha \downarrow 1} R(n, \alpha, p) \leq \frac{p}{1+p}, 
\]
and for any $(\alpha,p)$ 
\[
\displaystyle \limsup_{n \rightarrow \infty} R(n, \alpha, p) \leq \frac{p}{1+p}.
\]
\end{theorem}
\begin{proof}
Notice that from Corollary \ref{corollary_3}, $R(n, \alpha, p) < \frac{p}{1+p} \left[ 1+ \frac{\alpha^p (\alpha-1)}{\alpha^p + (n-1)}\right]$.
Thus we have
\[
\displaystyle \limsup_{\alpha \downarrow 1} R(n, \alpha, p) \leq
\lim_{\alpha \downarrow 1} \frac{p}{1+p} \left[ 1+ \frac{\alpha^p (\alpha-1)}{\alpha^p + (n-1)}\right] = \frac{p}{1+p} ,
\]
and
\[
\displaystyle \limsup_{n \rightarrow \infty} R(n, \alpha, p) \leq
\lim_{n \rightarrow \infty} \frac{p}{1+p} \left[ 1+ \frac{\alpha^p (\alpha-1)}{\alpha^p + (n-1)}\right] = \frac{p}{1+p} .
\]
\end{proof}


\section{Discussion}
\label{sec_conclusion}

This paper focuses on revenue-maximizing mechanism design for quasi-proportional auctions.
Specifically, for the general $n$ bidder case, we have proved the existence of the pure-strategy Nash equilibrium
and given lower bounds for bids at an equilibrium.
For the OLOS case, we have also (1) proved the uniqueness of the pure-strategy Nash equilibrium, and
(2) developed an approach to efficiently compute the equilibrium revenue.
We have also presented how to numerically solve the revenue-maximizing mechanism design problem in the 
OLOS case. 
We used computation to show that steeper weight functions maximize revenue when there is more competition, and we used analysis to show the importance of selecting $p$ based on $\alpha$ and, to a lesser extent, on $n$. 

In practice, the auctioneer does not know the precise value of $\alpha$ (and perhaps $n$).  However, the auctioneer may know that (with high probability), $(\alpha,n)$ lies in a some subset $\mathcal{D}$ of $(1,\infty) \times \{2,3,\ldots \}$.
Since $p^*$ is robust to small changes in both $\alpha$ and $n$ (see Figure~\ref{R_star_alpha}), if $\mathcal{D}$ is ``small", then we can choose
$p$ based on any $(\alpha,n) \in \mathcal{D}$.
If $\mathcal{D}$ is ``large", then we can choose an exponent $\tilde{p}$ by robust optimization:
\be
\tilde{p} \in \argmax_{p: p>0} \min_{(\alpha, n) \in \mathcal{D}} R(n, \alpha, p).
\ee
Since there are only three decision variables, this (non-convex) optimization problem can be solved numerically.

One challenge for further research is to prove that there is a unique pure-strategy Nash equilibrium for $p>1$. Another challenge is to produce a closed-form solution for equilibrium bids, based on bidders' private values. Analyzing equilibria for other weight functions, such as exponentials, would be an interesting way to extend this work. (Use $f(x) = e^x-1$ to avoid awarding an allocation for a bid of zero.)

It would also be interesting to explore strategies for the auctioneer's selections of $p$ over repeated auctions. In a model where bidders need some allocation in order to learn their private values, the auctioneer may benefit from starting with a lower $p$ value to have more of an even allocation as bidders learn, then increasing $p$ from auction to auction. As the bidders learn their private values, the auctioneer may gain information about the bidders' private values from their evolving bids. It would also be interesting to explore optimal settings for $p$ in a model where auctioneers compete against each other for bidders. 


\bibliographystyle{splncs03}
\bibliography{qp_wine_2015}

%


\newpage
\appendix
\begin{center}
\textbf{\Large Appendices}
\end{center}
\section{Proofs}
\subsection{Proof for Theorem \ref{thm_best_response}}

We will first prove a lemma about allocations and their derivatives, apply it to prove a lemma about weight functions, then apply that lemma to our weight functions. 

\begin{lemma} \label{lem1}
If $a a''  < 2 (a')^2$, then $(u' = 0) \Rightarrow (u''<0)$.
\end{lemma}

\begin{proof}[of Lemma \ref{lem1}]
Since
\be
u' = a'(v-b) - a,
\ee
$u'=0$ implies $a' \neq 0$ (otherwise $a=0$ also holds and $a a''  = 2 (a')^2$) and 
\be
v-b = \frac{a}{a'}. \label{u1}
\ee
Since 
\be
u'' = a''(v-b) - 2 a',
\ee
combining with Equation \ref{u1}, we have
\be
u'' = \frac{a a'' }{a'}  - 2 a' = \frac{1}{a'} \left[a a'' -  2 (a')^2 \right].
\ee
Since $a'>0$ and $a a''  < 2 (a')^2$, we have $u''<0$.
%
\end{proof}

\begin{lemma} \label{lem2}
If $f f'' < 2 (f')^2$, then $(u' = 0) \Rightarrow (u''<0)$.
\end{lemma}

\begin{proof}[of Lemma \ref{lem2}]
Note that
\be
a = \frac{f}{f+s},
\ee
\be
a' = \frac{f' s}{(f+s)^2},
\ee
and
\be
a'' = \frac{s[f'' (f+s) - 2 (f')^2]}{(f+s)^3}.
\ee
Substitute into $a a''  < 2 (a')^2$ from Lemma \ref{lem1} and multiply by $s (f+s)^4$:
\be
f^2 f'' + f f'' s - 2 f (f')^2 < 2 (f')^2 s.
\ee
\be
f^2 f'' + f f'' s < 2 f (f')^2 + 2 (f')^2 s.
\ee
\be
f f'' (f+s) < 2 (f')^2 (f+s).
\ee
\be
f f'' < 2 (f')^2.
\ee
\end{proof}

\begin{proof}[of Theorem \ref{thm_best_response}]
Note that 
\be
f = b^p, f'=pb^{p-1},  \hbox{and } f''=p(p-1)b^{p-2}.
\ee
Substitute into $f f'' < 2 (f')^2$ from Lemma \ref{lem2}:
\be
b^p p (p-1) b^{p-2} < 2 (p b^{p-1})^2.
\ee
\be
p (p-1) b^{2p-2} < 2 p^2 b^{2p-2}.
\ee
\be
p (p-1) < 2 p^2.
\ee
This holds for $p>0$.
\end{proof}

\subsection{Proof for Theorem \ref{thm_bounded_mapping}}

\begin{proof}
Consider any bidder $i$. To simplify the exposition, we drop the subscripts $i$ while we focus on that single bidder. Let $u$ be their utility function and $w$ be their lower bound $w_i$. Let $s$ be the sum of the weight function over other bidders' bids:
\be
s = \sum_{j \not= i} f(b_j),
\ee
and assume bids are at least their lower bounds: $\forall j \not= i: b_j \geq w_j$. Recall from Theorem \ref{thm_best_response} that the utility function $u$ has a single local maximum, so $u' \geq 0$ at $b=w$ implies that the best response is at least $w$. At $b=w$, 
\be
u' = a'(v-w) - a.
\ee
So 
\be
v - w \geq a / a'
\ee
implies $u'\geq 0$ at $b=w$. Substitute
\be
a = \frac{f}{f+s} \hbox{ and } a'=\frac{f's}{(f+s)^2}:
\ee
\be
v - w \geq \frac{f (f+s)^2}{(f+s) f' s}.
\ee
Cancel $f+s$ and do some algebra:
\be
w \leq v - \frac{f}{f'}(1+ \frac{f}{s}).
\ee
Substitute $f = w^p$ and $f' = p w^{p-1}$:
\be
w \leq v - \frac{w}{p}(1+ \frac{f}{s}).
\ee
Solve for w:
\be
w \leq \frac{v}{1 + \frac{1}{p} (1+ \frac{f}{s})}. \label{ww}
\ee
\end{proof}

\subsection{Proof for Lemma \ref{lemma_1}}

\begin{proof}
Similarly as the previous section, we define
$s_i^*=\sum_{j \neq i} f(b_j^*)$. From the first-order condition at a pure-strategy Nash equilibrium, we have\footnote{Notice that the results in the previous section indicate the denominator is bounded away from $0$.}
\begin{equation}
\frac{f(b_i^*) \left[f(b_i^*) +s_i^* \right]}{f'(b_i^*) s_i^* \left[ v_i -b_i^* \right]}=1 \quad \forall i=1,2, \ldots, n.  \label{eqn:first_order}
\end{equation}
Notice that $f(b_i^*) +s_i^*$ is a constant for all $i$.
Thus, for $f(x)=x^p$ and any $i,j \geq 2$, we have
$
\frac{b_i^*}{s_i^* \left[ 1- b_i^*\right]} = \frac{b_j^*}{s_j^* \left[ 1- b_j^*\right]}$ . 
Let $m=\sum_{k \neq i,j} f(b_k^*)$, we have 
$s_i^*=m +(b_j^*)^p$ and $s_j^*=m +(b_i^*)^p$, hence we have
\[
\frac{b_i^* \left[ m +(b_i^*)^p \right]}{ \left[ 1- b_i^*\right]} = \frac{b_j^*  \left[ m +(b_j^*)^p \right]}{ \left[ 1- b_j^*\right]}.
\]
Notice that $0< b_i^*, b_j^*<1$, thus, to prove $b_i^*=b_j^*$, it is sufficient to prove that function
$g(x)=\frac{x \left[ m +x^p\right]}{1-x}$ is strictly monotone in interval $(0,1)$, for any $m$ and $p$, notice that
\[
g'(x)=\frac{1}{(1-x)^2} \left[ m +x^p + p x^p (1-x) \right]>0.
\]
Hence $b_i^*=b_j^*$ and we have proved the lemma.
\end{proof}

\subsection{Proof for Lemma \ref{lemma:case_n_2}}

\begin{proof}
Notice that when $z \geq \alpha$. we have
\[
(1+p)(n-1) z^{p+1} \geq \alpha (1+p)(n-1) z^p
\]
and
\[
 (n-1)(n-2)(1+p)z \geq (1+p)(n-2)(n-1)\alpha.
\]
Thus, $h(z) \geq z^{2p+1} + (n-2) z^{p+1} - \alpha (n-1)$. Since $z \geq \alpha >1$ and $p>0$, we have $z^p>1$ and
$z^{2p+1}  > z^{p+1} > z \geq \alpha$. Thus, $h(z)>0$.

Similarly, when $z \leq \alpha^{\frac{1}{1+2p}}$, we have
\[
(1+p)(n-1) z^{p+1} < \alpha (1+p)(n-1) z^p
\]
and
\[
 (n-1)(n-2)(1+p)z \leq (1+p)(n-2)(n-1)\alpha.
\]
Thus, $h(z) < z^{2p+1} + (n-2) z^{p+1} - \alpha (n-1)$. Since $z \leq \alpha^{\frac{1}{1+2p}}$ and $p>0$, we have
$z^{1+2p} \leq \alpha$ and $z^{1+p} \leq \alpha^{\frac{1+p}{1+2p}}<\alpha$. Thus, $h(z)<0$.
\end{proof}

\subsection{Proof for Lemma \ref{lemma:case_n_3}}

\begin{proof}
Notice that from Lemma \ref{lemma:case_n_2}, the equation $h(z)=0$ has no solution in interval $\left( 0,  \alpha^{\frac{1}{2p+1}} \right]$ and interval
$\left[\alpha, \infty \right)$. Moreover, since function $h(z)$ is continuous, $h \left( \alpha^{\frac{1}{1+2p}}\right)<0$, and $h(\alpha)>0$, 
equation $h(z)=0$ has at least one solution in the interval $\left ( \alpha^{\frac{1}{2p+1}}, \alpha \right )$.

Note that to prove the uniqueness of the solution in the interval $\left ( \alpha^{\frac{1}{2p+1}}, \alpha \right )$, it is sufficient to prove that
$h'(z)>0$ when $h(z) \geq 0$. That is, to prove that $h(z)$ is \emph{strictly increasing} when it is \emph{nonnegative}.
To simplify the exposition, we write $h(z)$ as
\[
h(z)=z^{2p+1} + c_1 z^{p+1} - c_2 z^p + c_3 z - c_4,
\]
where $c_1 = (n-2) + (1+p)(n-1)$, $c_2= \alpha (1+p)(n-1)$, $c_3= (n-1)(n-2)(1+p)$, and
$c_4=\alpha (n-1) \left[ (1+p)(n-2)+1\right]$ (see Equation~\ref{eqn:h_equation}).
Notice that $c_1, c_2, c_4>0$, and $c_3 \geq 0$ (note $c_3=0$ when $n=2$). Hence we have
\[
h'(z) = (2p+1) z^{2p} + (p+1) c_1 z^p - p c_2 z^{p-1} +c_3.
\]

The first observation is that in interval $\left ( \alpha^{\frac{1}{2p+1}}, \alpha \right )$, we always have
\[
c_3 z -c_4 < (n-1)(n-2)(1+p) \alpha - \alpha (n-1) \left[ (1+p)(n-2)+1\right] = -\alpha (n-1) <0.
\]
Thus, $h(z) \geq 0$ implies that
\[
z^{2p+1} + c_1 z^{p+1} - c_2 z^p > \alpha (n-1) >0, 
\]
which leads to
\[
pz^{2p} + p c_1 z^{p} -p c_2 z^{p-1} > \alpha p (n-1)/ z > p(n-1) >0,
\]
where the second inequality follows from $z < \alpha$. Thus we have
\[
h'(z) = (2p+1) z^{2p} + (p+1) c_1 z^p - p c_2 z^{p-1} +c_3 >pz^{2p} + p c_1 z^{p} -p c_2 z^{p-1} +c_3 > p(n-1)+c_3 >0.
\]
Thus, in interval $\left ( \alpha^{\frac{1}{2p+1}}, \alpha \right )$, $h'(z)>0$ when $h(z) \geq 0$.
\end{proof}

\end{document}